\documentclass[11pt,a4paper]{article}

\usepackage[utf8]{inputenc}
\usepackage[T1]{fontenc}

\usepackage{amsmath,amssymb,amsthm}

\newtheorem{theorem}{Theorem}[section]
\newtheorem{proposition}[theorem]{Proposition}
\newtheorem{corollary}[theorem]{Corollary}

\theoremstyle{definition}
\newtheorem{definition}[theorem]{Definition}

\theoremstyle{remark}

\newtheorem{property}[theorem]{Property}

\DeclareMathOperator*{\argmax}{arg\,max}

\usepackage{booktabs}
\usepackage{graphicx}
\usepackage{float}

\graphicspath{{figures/}{./figures/}{../figures/}}

\usepackage{algorithm}
\usepackage{algpseudocode}

\usepackage{listings}
\usepackage{xspace}
\usepackage{xcolor}

\usepackage{hyperref}
\hypersetup{
    colorlinks=true,
    linkcolor=blue,
    citecolor=blue,
    urlcolor=blue
}
\usepackage[capitalise,noabbrev]{cleveref}

\usepackage{geometry}
\newcommand{\CVPL}{\mathrm{CVPL}}

\newcommand{\CVPLLR}{\text{CVPL-LR}}

\newcommand{\opB}{B_\lambda}              
\newcommand{\opphi}{\varphi}              
\newcommand{\oppsi}{\psi}                 
\newcommand{\optau}{\tau}                 

\newcommand{\Dor}{D}                      
\newcommand{\Dpr}{D'}                     

\newcommand{\Lstar}{\mathcal{L}^*}        
\newcommand{\ystar}{y^*}                  

\newcommand{\Qset}{Q}                     
\newcommand{\Aset}{A}                     
\newcommand{\Sset}{S}                     

\newcommand{\Bset}{\mathcal{B}}           
\newcommand{\Bcand}[1]{C_B(#1)}           
\newcommand{\Rblock}{R_{\text{block}}}    

\newcommand{\TLR}{\text{TLR}}             
\newcommand{\FLR}{\text{FLR}}             
\newcommand{\PrecOne}{\text{Prec@1}}      

\newcommand{\RiskSurf}[2]{R(#1, #2)}      

\newcommand{\simcos}{\mathrm{cos}}

\newcommand{\Splus}{S^{+}}                
\newcommand{\Sminus}{S^{-}}               

\newcommand{\Reals}{\mathbb{R}}

\newcommand{\Ind}{\mathbb{1}}             
\newcommand{\Prob}{\mathbb{P}}            
\newcommand{\Exp}{\mathbb{E}}             

\newcommand{\vx}{\mathbf{v}_x}

\newcommand{\zx}{\mathbf{z}_x}
\newcommand{\zy}{\mathbf{z}_y}

\newcommand{\eg}{e.g.\@\xspace}
\newcommand{\ie}{i.e.\@\xspace}

\hypersetup{
    pdfauthor={Valery Khvatov, Alexey Neyman},
    pdftitle={CVPL: A Geometric Framework for Post-Hoc Linkage Risk Assessment in Protected Tabular Data},
    pdfsubject={Privacy, Record Linkage, Synthetic Data},
    pdfkeywords={privacy, linkage risk, k-anonymity, synthetic data, PCA}
}

\begin{document}


\title{CVPL: A Geometric Framework for Post-Hoc Linkage Risk Assessment in Protected Tabular Data}

\author{
  Valery Khvatov\textsuperscript{1} \and 
  Alexey Neyman\textsuperscript{2}
}

\date{}

\maketitle

\begin{center}
\textsuperscript{1}REALM AI, Toronto, Canada \\
\texttt{valery.khvatov@realmdata.io} \\[0.5em]
\textsuperscript{2}The Big Data Association (Remote) / International
\end{center}

\begin{abstract}
Formal privacy metrics provide compliance-oriented guarantees but often
fail to quantify actual linkability in released datasets. Model-based
attacks, such as membership inference, can estimate empirical risk;
however, they are computationally expensive and typically require access
to trained models. A practical, data-centric tool for measuring residual
linkage risk, therefore, remains lacking.

We introduce \textbf{CVPL (Cluster-Vector-Projection Linkage)}, a
geometric framework for post-hoc assessment of linkage risk between
original and protected tabular data. CVPL represents linkage analysis as
an operator pipeline comprising blocking, vectorization, latent
projection, and similarity evaluation, yielding continuous,
scenario-dependent risk estimates rather than binary compliance
verdicts. We define CVPL formally under an explicit threat model and
introduce threshold-aware risk surfaces $R(\lambda, \tau)$ that
capture the joint effects of protection strength and attacker
strictness. We further establish a progressive blocking strategy with
monotonicity guarantees, enabling anytime risk estimation with valid
lower bounds (convergence from below).

We show that classical Fellegi--Sunter linkage arises as a special case
of CVPL under restrictive assumptions, and that violations of these
assumptions can induce systematic over-linking bias. Empirical
validation on \emph{10,000 records across 19 protection} configurations
demonstrates that formal k-anonymity compliance may coexist with
substantial empirical linkability, a significant portion of which arises
from non--quasi-identifier behavioral patterns rather than traditional
quasi-identifiers. Finally, CVPL provides interpretable diagnostics
identifying which features drive linkage feasibility.

CVPL does not offer formal privacy guarantees. Instead, it supports
privacy impact assessment, protection mechanism comparison, and
\emph{utility--risk trade-off} analysis.
\end{abstract}



\section{Introduction}\label{introduction}

\subsection{Motivation}\label{motivation}

Data anonymization and synthetic data generation are widely used to
support data sharing, analytics, and machine learning while reducing
privacy risks. Formal privacy criteria such as k-anonymity \cite{sweeney2002},
l-diversity \cite{machanavajjhala2007}, t-closeness \cite{li2007}, and differential privacy
\cite{dwork2006calibrating} provide strong, compliance-oriented guarantees. However, they
often offer limited visibility into empirical risks that arise when
protected data continues to preserve meaningful statistical dependence
and geometric structure.

In practice, satisfying a formal criterion does not imply low
linkability under realistic attacker assumptions. A dataset that
achieves k = 20 anonymity may still be linkable if correlations across
attributes, temporal regularities, or latent cluster structures are
retained for utility. Formal criteria answer the question, ``Is the
release compliant?'', but they do not address the complementary
question, ``What is the residual linkage risk under plausible attack
scenarios?''

Model-based privacy attacks provide a different perspective. Membership
inference, for example \cite{shokri2017membership}, estimates empirical leakage by attacking
a trained model. However, this class of methods is not designed for
post-hoc evaluation of released datasets. It typically requires access
to trained models, significant computational resources, and produces
outputs that are difficult to interpret in governance, audit, and
privacy impact assessment workflows.

These gaps motivate the need for a practical, data-centric approach that
measures residual linkage risk directly on released datasets, under an
explicit threat model, and with outputs that remain interpretable for
privacy assessment and decision-making.

\subsection{The Gap}\label{the-gap}

Current approaches to privacy risk assessment occupy distinct but weakly
connected positions.

\begin{table}[htbp]
\centering
\caption{Comparison of privacy risk assessment approaches.}
\label{tab:approaches}
\begin{tabular}{@{}lp{3.5cm}p{5.5cm}@{}}
\toprule
Approach & Strength & Limitation \\
\midrule
Formal privacy metrics & Rigorous guarantees & No empirical measurement of linkage risk \\
Model-based attacks & Empirical risk estimation & Requires model access and high computational cost \\
Classical record linkage & Established methodology & Distance-based metrics \\
Distance-based metrics & Simple, model-free quantification & May be misleading under correlation and lack of structural context \\
\bottomrule
\end{tabular}
\end{table}

Formal privacy metrics certify compliance but do not quantify residual
linkability in released data. Model-based attacks estimate leakage
empirically, yet they target trained models rather than datasets and are
poorly suited for post-hoc audits. Classical record linkage provides
mature techniques but assumes attribute-level independence and focuses
on deterministic or binary matching outcomes. Distance-based metrics
offer simplicity and scalability but fail to capture higher-order
dependence and latent structure preserved by utility-oriented
transformations.

This work addresses this methodological gap by proposing \textbf{CVPL},
a framework designed not to recover identities or perform
de-anonymization, but to measure residual linkage risk in protected
datasets under explicitly defined threat models.

\subsection{Central Thesis}\label{central-thesis}

\emph{CVPL introduces a geometric, simulation-based approach to linkage
risk assessment that complements formal privacy criteria by quantifying
residual risk arising from preserved data structure. By exploiting
structural invariants that protection mechanisms retain for utility,
CVPL reveals linkability that may persist even under formal compliance
and reframes privacy as a continuous, measurable quantity rather than a
binary property.}

\subsection{Contributions}\label{contributions}

\textbf{C1. Framework Formalization.} We define CVPL as an explicit
operator composition,

\[\text{CVPL} = \tau \circ s \circ \psi \circ \varphi \circ B,\]

providing a clear separation between linkage mechanics, evaluation
protocol, and threat assumptions.

\textbf{C2. Threat Model Specification.} We formalize the attacker model
underlying CVPL and distinguish it from cooperative record linkage and
model-based privacy attacks, positioning CVPL as a post-hoc risk
assessment framework rather than an attack method.

\textbf{C3. Risk Surface Perspective.} We argue that linkage risk should
be characterized as a surface \(R(\lambda,\tau)\)over protection
parameters and attacker strictness, avoiding arbitrary single-threshold
assessments.

\textbf{C4. Mathematical Properties.} We establish Theorem 1 showing
that blocking relaxation monotonically increases risk estimates,
enabling anytime evaluation with guaranteed lower bounds.

\textbf{C5. Empirical Validation.} We establish key properties of the
framework, including the decomposition of blocking-induced bias,
threshold-based false-positive control, the relationship to
Fellegi--Sunter linkage under restrictive assumptions, and FS
over-linking bias under representation shift.

\textbf{C6. Practical Applicability.} We evaluate CVPL on a realistic
marketing touchpoint scenario with 10,000 records across 19 protection
configurations, showing that formal k-anonymity compliance coexists with
substantial empirical linkability, with 60.4\% of detected risk arising
from non-QI behavioral patterns.

\textbf{C7. Interpretable Diagnostics.} We provide feature-level
attribution showing which attributes drive linkability, enabling
actionable insights for privacy engineering.

\textbf{C8. Baseline Comparison.} We systematically compare CVPL against
Fellegi-Sunter, DCR, and NNDR, demonstrating complementary strengths and
identifying systematic biases in existing methods.

\subsection{Scope and Limitations}\label{scope-and-limitations}

CVPL is an empirical risk estimator, not a formal privacy guarantee. It
measures linkage feasibility under specified assumptions and should be
interpreted as one component of a comprehensive privacy assessment, not
a replacement for formal analysis. Key limitations include: (1)
validation requires ground truth linkages, typically available only in
simulation or controlled settings; (2) results depend on feature
selection, blocking strategy, and projection method; (3) experiments in
this paper use simulated data - extension to real datasets is discussed
but not fully demonstrated. We discuss limitations explicitly in Section
10.

\subsection{Paper Organization}\label{paper-organization}

Section 2 defines the problem setting and threat model. Section 3
reviews related to work. Section 4 formalizes the CVPL framework,
including progressive blocking. Section 5 presents the evaluation
protocol. Section 6 introduces threshold calibration and risk surfaces.
Section 7 describes experimental design. Section 8 presents results
including baseline comparisons and convergence analysis. Section 9
discusses implications. Section 10 addresses limitations. Section 11
concludes.


\section{Problem Setting and Threat Model}\label{sec:problem-setting}

\subsection{Data Model}\label{sec:data-model}

Let:
\begin{itemize}
    \item $\Dor = \{x_i\}_{i=1}^{n}$ denote the \textbf{original (source) dataset},
    \item $\Dpr = \{y_j\}_{j=1}^{m}$ denote the \textbf{protected dataset}, obtained from $\Dor$ through anonymization, perturbation, or synthetic data generation:
    \begin{itemize}
        \item Record-preserving transformations (k-anonymity, perturbation): $m = n$ with implicit correspondence between $x_i$ and $y_i$.
        \item Generative transformations (synthetic data): $m$ may differ from $n$, with no guaranteed correspondence between individual records.
    \end{itemize}
    \item $\Lstar \subset \Dor \times \Dpr$ denote the \textbf{true linkage relation}, used exclusively for evaluation and never available to the attacker or the CVPL procedure.
\end{itemize}

Each record is represented as a tuple of attributes, partitioned into three disjoint subsets:
\begin{itemize}
    \item $\Qset$: \textbf{quasi-identifiers}, observable by the attacker and used for blocking and candidate generation.
    \item $\Aset$: \textbf{analytical attributes}, used for vectorization and similarity computation.
    \item $\Sset$: \textbf{sensitive attributes}, which constitute the protection target and are not directly used in CVPL.
\end{itemize}

This partition reflects a common privacy threat setting in which attackers can observe generalized identifiers and analytical features, while sensitive attributes are protected or perturbed.

\subsection{Threat Model}\label{sec:threat-model}

CVPL considers an attacker with the following capabilities:

\begin{enumerate}
    \item \textbf{Full access to released datasets.} Full access to the protected dataset $\Dpr$ and (potentially) to an auxiliary dataset $X$ containing records about the same or overlapping population, with shared attributes in $\Qset \cup \Aset$.
    
    \item \textbf{No mechanism access.} No knowledge of anonymization parameters, synthetic generator architecture, or noise distributions.
    
    \item \textbf{Quasi-identifier awareness.} Ability to identify and exploit generalized quasi-identifiers for blocking.
    
    \item \textbf{Similarity-based inference.} Capability to perform vectorization, projection, and similarity analysis within candidate blocks.
    
    \item \textbf{Variable strictness.} Ability to vary the decision threshold $\optau$ according to risk tolerance or attack strategy.
    
    \item \textbf{Worst-case auxiliary data.} For risk assessment, CVPL sets $X = \Dor$ (the original dataset). This conservative assumption provides an upper bound on linkage risk: any realistic $X$ with partial population overlap or incomplete attributes yields a risk bounded above by the CVPL estimate. \Cref{sec:progressive-blocking} formalizes this through monotonicity under relaxation operators.
\end{enumerate}

CVPL explicitly does not assume:
\begin{itemize}
    \item access to hidden identifiers or ground-truth linkage during the attack,
    \item availability of trained models or generation mechanisms,
    \item knowledge of the specific auxiliary dataset $X$ that any real attacker may possess (risk is evaluated under a worst-case envelope).
\end{itemize}

This threat model reflects a realistic post-hoc auditing scenario rather than an adaptive or interactive attack, while providing conservative risk estimates applicable across a range of attacker capabilities.

\subsection{Attack Goal}\label{sec:attack-goal}

The attacker's objective is to determine, for records $x \in \Dor$, whether there exist records $y \in \Dpr$ such that the pair $(x, y)$ plausibly corresponds to the same underlying entity.

Success is measured in terms of linkage feasibility, defined as the existence of statistically plausible candidate links under the attacker's similarity and threshold assumptions. CVPL does not aim to recover unique identities or produce definitive matches, but rather to assess whether protected data permits meaningful linkage at all.

\subsection{Distinction from Related Attack Models}\label{sec:attack-distinction}

\begin{table}[htbp]
\centering
\caption{Comparison of attack types and their relation to CVPL.}
\label{tab:attack-types}
\begin{tabular}{@{}llll@{}}
\toprule
\textbf{Attack Type} & \textbf{Target} & \textbf{Access Required} & \textbf{Relation to CVPL} \\
\midrule
Membership inference & ML model & Model queries & Complementary, model-side \\
Attribute inference & ML model & Model and partial record & Different target \\
Cooperative record linkage & Databases & Both datasets & Similar techniques, different intent \\
\textbf{CVPL} & Data release & Datasets only & Adversarial risk assessment \\
\bottomrule
\end{tabular}
\end{table}

CVPL is best understood as adversarial record linkage for privacy auditing. It applies linkage techniques not to integrate databases, but to quantify residual linkage risk in protected data releases.

\subsection{Linkage Semantics: Existential vs.\ Unique}\label{sec:linkage-semantics}

We distinguish two notions of linkage corresponding to different privacy questions.

\paragraph{Existential linkage.}
A record $x \in \Dor$ is considered linkable if there exists at least one candidate $y \in \Dpr$ such that the similarity exceeds a decision threshold $\optau$ (with $s(x,y) \in [0,1]$ and $\optau \in [0,1]$):
\begin{equation}\label{eq:existential-linkage}
    \exists\, y \in \Dpr \;\text{such that}\; s(x,y) \geq \optau.
\end{equation}

The threshold $\optau$ controls the trade-off between sensitivity and specificity: lower values increase recall (\ie more records deemed linkable) at the cost of precision, while higher values yield more conservative estimates. This notion captures a privacy-relevant worst-case risk. It reflects whether an attacker can plausibly claim that a link exists, even when multiple candidates satisfy the criterion.

\paragraph{Unique linkage (top-1).}
A record $x \in \Dor$ is uniquely linkable if the highest-similarity candidate corresponds to the true match:
\begin{equation}\label{eq:unique-linkage}
    \argmax_{y \in \Dpr} s(x,y) \in L^*(x),
\end{equation}
where $L^*(x)$ denotes the true counterpart of $x$ in $\Dpr$, when ground truth is available for evaluation. For protection mechanisms that may produce multiple valid counterparts (\eg one-to-many synthetic generation), we generalize to $L^*(x) \subseteq \Dpr$ and require $\argmax \in L^*(x)$. This notion measures de-anonymization accuracy rather than linkage feasibility.

By default, \textbf{CVPL-LR (linkage rate)} is defined using existential linkage, as it answers the privacy-relevant question: \emph{Can a link be plausibly claimed under the attacker model?} In scenarios where unique identification accuracy is required, we report top-1 metrics, such as top-1 precision, separately.

\paragraph{Formal clarification.}
CVPL-LR measures the probability of existential linkability, not the probability of correct re-identification. To formalize this, we introduce two components:

\begin{itemize}
    \item Let $B\colon \Dor \cup \Dpr \to \Bset$ denote a \textbf{blocking function} derived from quasi-identifiers $\Qset$ (\eg exact match, hierarchical grouping, or locality-sensitive hashing), mapping each record to a discrete block identifier. For a record $x \in \Dor$, the \textbf{candidate block} is:
    \begin{equation}\label{eq:candidate-block}
        \Bcand{x} = \{y \in \Dpr : B(y) = B(x)\}.
    \end{equation}
    The full specification of blocking mechanics appears in \cref{sec:blocking}.
    
    \item Let $x \sim \mathcal{U}(\Dor)$ denote a record sampled uniformly at random from $\Dor$.
\end{itemize}

The linkage rate is then defined as the probability that a randomly sampled record has at least one candidate in its block exceeding the similarity threshold:
\begin{equation}\label{eq:cvpl-lr}
    \CVPLLR(\optau) = \Prob_{x \sim \mathcal{U}(\Dor)}\bigl[\exists\, y \in \Bcand{x} : s(x,y) \geq \optau\bigr].
\end{equation}

This notion is strictly weaker than identifying the true match $\ystar$. Accordingly, CVPL-LR should be interpreted as a \emph{conservative risk exposure metric} capturing the existence of plausible links under the chosen similarity and blocking model, rather than as an attack success rate for correct re-identification. A high CVPL-LR indicates that plausible claims of linkage exist even if the attacker cannot uniquely identify the correct target.

\paragraph{Implication for privacy assessment.}
This existential interpretation aligns with regulatory and risk-based perspectives in which the possibility of linkage itself constitutes privacy risk, independent of whether correct identification is ultimately achieved.


\section{Related Work}\label{sec:related-work}

\subsection{Formal Privacy Criteria}\label{sec:formal-privacy-criteria}

\paragraph{k-anonymity and extensions.}
k-anonymity was introduced by Sweeney~\cite{sweeney2002} and requires each record to be indistinguishable from at least $(k-1)$ others with respect to a set of quasi-identifiers. Well-known extensions include $\ell$-diversity~\cite{machanavajjhala2007} and $t$-closeness~\cite{li2007}, which attempt to address attribute disclosure by constraining the distribution of sensitive values within equivalence classes.

These criteria provide compliance-oriented conditions that are typically evaluated as binary pass or fail outcomes for a given set of quasi-identifiers. However, they do not quantify residual linkage risk that can arise when non-QI attributes preserve identity-correlated structure. In practice, a dataset can satisfy k-anonymity while still exhibiting substantial linkability due to retained correlations, temporal regularities, or cluster structure needed for utility.

\paragraph{Differential privacy.}
Differential privacy (DP) provides a formal bound on the influence of an individual on an algorithm's outputs~\cite{dwork2006calibrating}. DP was originally formulated for query answering and statistical releases, and it has since been applied to generative modeling and DP-based synthetic data generation. In data-release settings, DP-based synthesis offers theoretical guarantees, but it does not directly produce an interpretable estimate of linkage risk for a specific released dataset under a concrete attacker model.

CVPL complements these approaches by providing a data-centric, post-hoc measurement framework that helps determine whether formal guarantees align with empirical linkability in practice.

\subsection{Model-Based Privacy Attacks}\label{sec:model-based-attacks}

\paragraph{Membership inference.}
Membership inference attacks (MIA) demonstrate that trained machine learning models can leak information about whether a specific record was included in the training data~\cite{shokri2017membership}. This line of work has since expanded to broader empirical privacy tests targeting learned models, including generative models, which have been shown to memorize and, in some cases, reproduce training examples~\cite{carlini2019secret,carlini2021extracting}.

A key distinction lies in the attack surface. MIA primarily targets deployed models, whereas CVPL targets released datasets. As a result, the two approaches address different privacy questions and rely on different access assumptions.

\begin{table}[htbp]
\centering
\caption{Comparison of Membership Inference Attacks (MIA) and CVPL.}
\label{tab:mia-vs-cvpl}
\begin{tabular}{@{}lp{4cm}p{5.5cm}@{}}
\toprule
\textbf{Aspect} & \textbf{MIA} & \textbf{CVPL} \\
\midrule
Attack object & Trained model (classifier, generator) & Released dataset \\
Question & ``Was $x$ in the training set?'' & ``Can $x$ be plausibly linked to some $y$?'' \\
Required access & Query or model access & Access to $\Dor$ and $\Dpr$ \\
Typical cost & High (shadow models, training) & Medium (blocking, linear algebra, similarity search) \\
Applies when & Model is deployed & Data is released \\
Mechanism knowledge & Often helpful or required & Not required \\
\bottomrule
\end{tabular}
\end{table}

CVPL is applicable in scenarios where model-based testing is unavailable or impractical, for example, when the protection mechanism is proprietary, the generator is not accessible, or the release is provided without a model interface.

Importantly, CVPL and MIA are complementary rather than competing. They measure different privacy risk surfaces, and results from one should not be interpreted as subsuming or dominating the other. CVPL is not intended to replace model-based attacks, but to provide a data-centric risk assessment when only released datasets are available.

\subsection{Classical Record Linkage}\label{sec:classical-linkage}

\paragraph{Fellegi--Sunter framework.}
Probabilistic record linkage was formalized by Fellegi and Sunter~\cite{fellegi1969theory} using attribute-level agreement patterns and likelihood ratio weighting, typically under conditional independence assumptions. The framework is foundational and remains widely used in data integration and entity resolution.

For privacy auditing, however, classical record linkage methods face several limitations. First, the independence assumptions underlying attribute-level weighting are often violated in protected data releases that intentionally preserve correlations to retain utility. Second, the design goal of classical linkage is usually cooperative integration rather than adversarial assessment under an explicit threat model. Third, many implementations focus on binary match or non-match decisions, whereas privacy auditing benefits from continuous risk measures and attacker-dependent regimes. Finally, performance may degrade under noise injections and synthetic transformations that alter marginal distributions while preserving latent structure.

In \cref{sec:framework}, we show that the Fellegi--Sunter linkage can be interpreted as a special case of the CVPL formulation under restrictive assumptions on representation and decision structure. In \cref{sec:results}, we demonstrate empirically that when these assumptions are violated, as is typical in utility-preserving protection, FS exhibits systematic over-linking bias (93\% linkage rate but only 13\% precision).

\subsection{Empirical Disclosure Risk}\label{sec:empirical-disclosure-risk}

CVPL is closely aligned with the tradition of \textbf{Empirical Disclosure Risk (EDR)} in statistical disclosure control, which emphasizes operational linkability---the practical feasibility of re-identification---rather than formal, worst-case privacy guarantees~\cite{elliot2010functional}. Within this perspective, privacy risk is assessed based on whether released data can be plausibly linked to external information under realistic assumptions, rather than solely on compliance with abstract criteria.

A substantial body of prior work has developed empirical approaches to disclosure risk assessment. \cite{skinner1994disclosure} provide comprehensive guidelines for data de-identification that rely on empirical risk evaluation under plausible attacker models. \cite{skinner1994disclosure} formalize disclosure risk measures for microdata, focusing on re-identification feasibility. The $\mu$-ARGUS framework~\cite{hundepool2012statistical} operationalizes empirical disclosure risk assessment in official statistics, combining record linkage techniques with practical auditing workflows.

CVPL extends this line of work in several respects. First, it models linkage risk in a geometric similarity space rather than relying solely on attribute-level agreement patterns. Second, it characterizes risk as a continuous surface over protection parameters and attacker thresholds, rather than as a single point estimate. Third, it explicitly incorporates attacker strictness as a first-class variable, enabling scenario-dependent risk analysis.

In this sense, CVPL should be viewed as a geometric and simulation-based instantiation of empirical disclosure risk assessment, designed for modern protected and synthetic datasets that preserve complex statistical structure.

\subsection{Synthetic Data Evaluation}\label{sec:synthetic-evaluation}

Recent work has emphasized that the use of synthetic data does not automatically guarantee privacy protection. \cite{stadler2022synthetic} demonstrate that synthetic data generators can leak training data through membership inference attacks, challenging the assumption of ``privacy by synthesis''. Related studies show that memorization and partial reproduction of training records can occur even when models are trained without explicit overfitting signals~\cite{carlini2019secret,carlini2021extracting,jordon2022synthetic,xu2019modeling}.

In practice, synthetic data evaluation often relies on distance-based heuristics, most notably Distance to Closest Record (DCR) and Nearest Neighbor Distance Ratio (NNDR)~\cite{zhao2021ctab}. While these metrics are simple to compute, they exhibit several limitations for privacy assessment. First, they measure geometric proximity rather than operational linkability and do not require ground-truth correspondence. Second, low distances may reflect preserved distributional structure or feature correlations rather than individual re-identification risk, leading to misleading interpretations~\cite{jordon2022synthetic}. Third, the choice of safety thresholds for such distances is inherently ambiguous and context-dependent.

Model-based evaluations, such as membership inference against generators, provide empirical leakage estimates but require access to the generative model and substantial computational effort. These approaches are therefore inapplicable when the generator is proprietary, inaccessible, or absent from the data release.

CVPL provides a complementary evaluation perspective for synthetic data. Rather than focusing on generator behavior or raw distance metrics, it assesses linkage feasibility directly at the dataset level by analyzing structural invariants preserved in the released data. This makes CVPL applicable in post-hoc auditing scenarios where only the original dataset $\Dor$ and the released dataset $\Dpr$ are available, and where the privacy risk arises from retained structure rather than explicit memorization.

\subsection{Positioning of CVPL}\label{sec:positioning}

CVPL occupies an intermediate methodological position within the landscape of privacy risk assessment techniques. It is more informative than binary compliance checks, as it estimates residual linkage risk under an explicit threat model rather than providing pass/fail guarantees. At the same time, it is typically less computationally demanding and less dependent on access to training procedures or internal mechanisms than model-based privacy attacks.

Unlike classical record linkage methods, which are primarily designed for cooperative data integration and entity resolution, CVPL is explicitly formulated for adversarial privacy auditing. Its objective is not to recover identities or establish definitive matches, but to quantify the feasibility of linkage under attacker-dependent assumptions. This leads to a focus on continuous, scenario-dependent risk characterization rather than binary match decisions.

Viewed across prior approaches discussed in \cref{sec:model-based-attacks,sec:classical-linkage,sec:empirical-disclosure-risk,sec:synthetic-evaluation}, CVPL bridges formal privacy metrics, empirical disclosure risk assessment, and classical linkage analysis. It combines the interpretability and dataset-level focus of empirical disclosure risk methods with a geometric representation that captures structural invariants preserved for utility. As a result, CVPL enables fine-grained risk analysis in settings where models are inaccessible, generators are proprietary, or only released datasets are available.

Conceptually, CVPL can be positioned between formal criteria and model-based attacks along three dimensions: computational cost, interpretability, and granularity of risk estimates. It offers higher resolution than distance-based heuristics (such as DCR and NNDR) and compliance metrics, while remaining more practical for post-hoc auditing than model-centric empirical attacks.


\section{CVPL Framework and Formal Definition}\label{sec:framework}

\subsection{Operator Decomposition}\label{sec:operator-decomposition}

We define CVPL as a composition of interpretable operators acting on a pair of datasets $(\Dor, \Dpr)$. Each record is partitioned into three disjoint attribute subsets: $\Qset$ (quasi-identifiers), $\Aset$ (observable analytical attributes), and $\Sset$ (sensitive attributes). All linkage operators act only on the observable projection $\Qset \cup \Aset$; $\Sset$ is excluded and used only for downstream inference analysis.

For parameters $\lambda$ (\emph{blocking configuration}) and $\optau$ (\emph{decision threshold}), we define:
\begin{equation}\label{eq:cvpl-composition}
    \CVPL_{\lambda,\optau} = \optau \circ s \circ \oppsi \circ \opphi \circ \opB.
\end{equation}

Here:
\begin{itemize}
    \item $\opB\colon \Dor \cup \Dpr \to \Bset$ is a blocking operator parameterized by $\lambda$, mapping records to block identifiers shared between $\Dor$ and $\Dpr$. The block identifier is derived from quasi-identifiers $\Qset$ (\eg exact match, hierarchical grouping, or locality-sensitive hashing). We adopt the term ``blocking'' from record linkage literature, where it denotes candidate pair reduction before detailed comparison---distinct from clustering or entity resolution.
    
    \item $\opphi\colon x \mapsto \vx \in \Reals^d$ is a vectorization operator that encodes the observable attributes ($\Qset \cup \Aset$) into a numeric feature space. Sensitive attributes $\Sset$ are excluded from vectorization, as they constitute the protection target.
    
    \item $\oppsi\colon \Reals^d \to \Reals^k$ with $k \ll d$ is a \emph{latent projection} (\eg PCA, UMAP), fit jointly on $\Dor \cup \Dpr$, producing compact embeddings.
    
    \item $s\colon \Reals^k \times \Reals^k \to \Reals$ is a similarity function, computed between latent embeddings.
    
    \item $\optau \in \Reals$ is a decision threshold controlling the attacker's strictness: larger $\optau$ yields fewer but stronger candidate links.
\end{itemize}

This decomposition isolates the contribution of each stage of the linkage process. Each operator corresponds to a distinct attacker capability or assumption, enabling interpretability, ablation analysis, and systematic variation of threat models.

\Cref{fig:cvpl-pipeline} illustrates the CVPL operator pipeline. For a record $x \in \Dor$, blocking restricts candidate matches to:
\begin{equation}\label{eq:candidate-set}
    C(x) = \{y \in \Dpr : \opB(x) = \opB(y)\}.
\end{equation}
A record is considered linkable if there exists at least one candidate whose similarity exceeds the threshold:
\begin{equation}\label{eq:linkable-indicator}
    \mathrm{LINKABLE}(x) = \Ind\bigl[\exists\, y \in C(x) : s(\oppsi(\opphi(x)), \oppsi(\opphi(y))) \geq \optau\bigr].
\end{equation}
The CVPL linkage rate (CVPL-LR) is then defined as the fraction of linkable records in $\Dor$:
\begin{equation}\label{eq:cvpl-lr-def}
    \CVPLLR(\optau) = \frac{|\{x \in \Dor : \mathrm{LINKABLE}(x) = 1\}|}{|\Dor|}.
\end{equation}

\begin{figure}[htbp]
    \centering
    \includegraphics[width=0.95\textwidth]{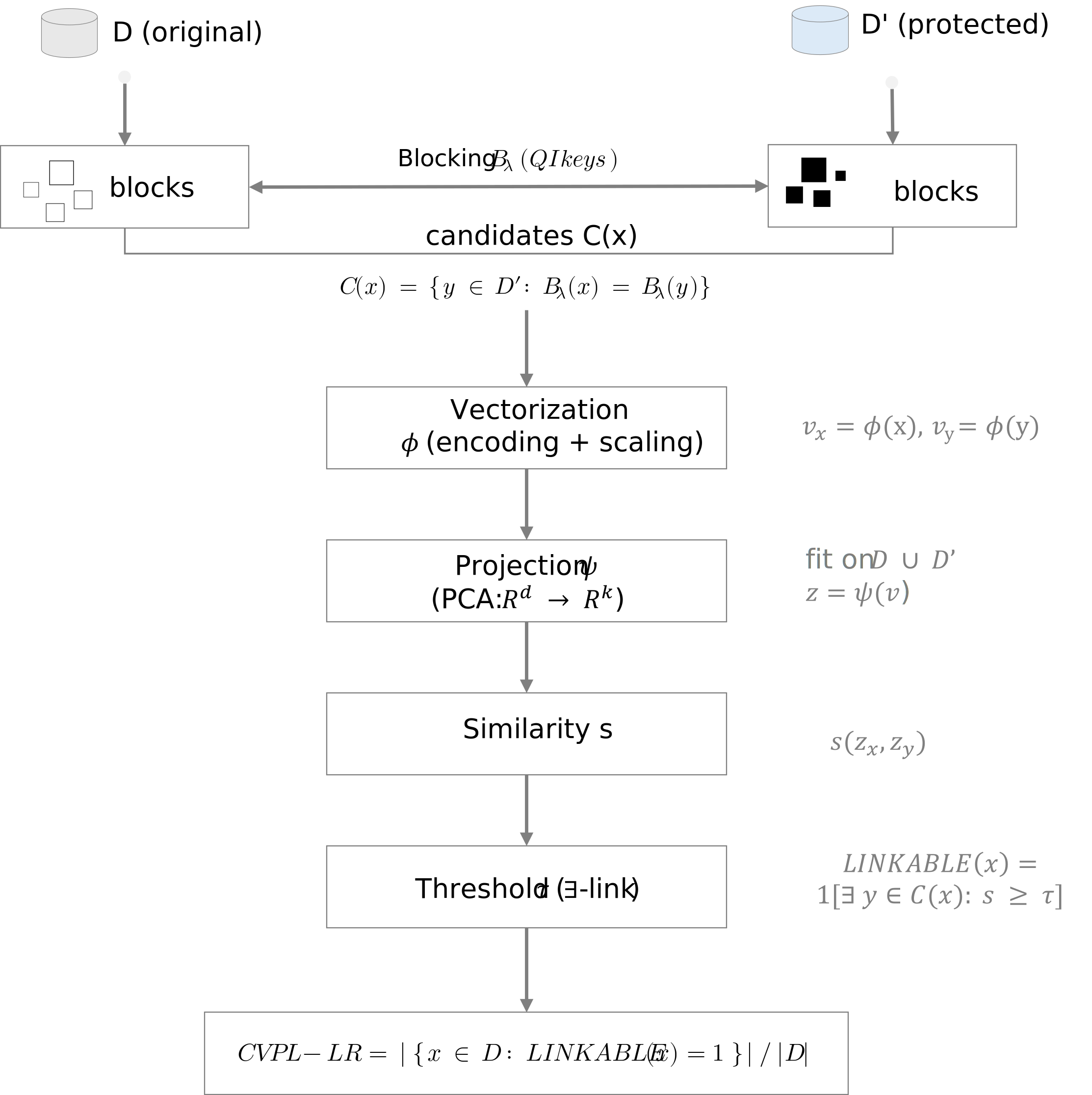}
    \caption{CVPL operator pipeline: blocking restricts candidates, vectorization and projection create embeddings, similarity scoring identifies potential links.}
    \label{fig:cvpl-pipeline}
\end{figure}

\subsection{Blocking Operator}\label{sec:blocking}

\paragraph{Definition.}
The blocking operator assigns each record to a block identifier based on generalized quasi-identifiers:
\begin{equation}\label{eq:blocking-def}
    \opB(x) = b \in \Bset,
\end{equation}
where $\Bset$ is a discrete set of block identifiers (\eg strings, integers, or tuples derived from quasi-identifier values). The cardinality $|\Bset|$ depends on blocking granularity $\lambda$: coarser blocking yields fewer, larger blocks (higher recall, lower precision), while finer blocking yields more, smaller blocks (lower recall, higher precision). In the extreme cases, $|\Bset| = 1$ places all records in a single block (full pairwise comparison), while $|\Bset| = |\Dor \cup \Dpr|$ isolates each record (no comparison).

A pair $(x, y)$ with $x \in \Dor$ and $y \in \Dpr$ is considered a candidate pair \emph{if and only if}:
\begin{equation}\label{eq:candidate-condition}
    \opB(x) = \opB(y).
\end{equation}

\paragraph{Structural recall.}
We define the blocking recall as:
\begin{equation}\label{eq:structural-recall}
    \Rblock = \Prob\bigl[\opB(x) = \opB(\ystar) \mid (x, \ystar) \in \Lstar\bigr],
\end{equation}
where $\Lstar$ denotes the set of true matching pairs. This quantity bounds the maximum true positive rate achievable by any linkage method that respects the blocking constraint, regardless of the choice of representation or similarity function.

\paragraph{Relationship to the protection mechanism.}
The blocking function $\opB$ may, but need not, coincide with the generalization function used by the protection mechanism. Three cases arise:
\begin{itemize}
    \item \textbf{Aligned blocking:} $\opB$ uses the same generalization as the protection (\eg identical age bins as k-anonymity). All true pairs are guaranteed to share a block, yielding structural recall $\Rblock = 1$.
    
    \item \textbf{Coarser blocking:} $\opB$ uses broader categories than the protection. True pairs remain in the same block ($\Rblock = 1$), but candidate pools are larger.
    
    \item \textbf{Finer blocking:} $\opB$ uses narrower categories than the protection. Some true pairs may be separated into different blocks ($\Rblock < 1$), but candidate pools are smaller.
\end{itemize}

\paragraph{Blocking as a structural risk constraint.}
Blocking is not merely computational optimization. It imposes a structural constraint on the attacker and defines an upper bound on the linkage that can be achieved by any similarity-based strategy operating under the same blocking scheme. Once records are separated into distinct blocks, no downstream representation or similarity function can recover links across blocks.

\paragraph{Implementation.}
In practice, blocking is typically implemented using generalized quasi-identifiers, for example: age binned into 5- or 10-year intervals, region mapped to hierarchical levels such as country, province, or city, and temporal attributes rounded to coarse intervals.

\subsection{Vectorization Operator}\label{sec:vectorization}

\paragraph{Definition.}
The vectorization operator maps each record to a numerical feature representation:
\begin{equation}\label{eq:vectorization-def}
    \opphi\colon x \mapsto \vx = (\varphi_1(x), \ldots, \varphi_d(x)).
\end{equation}
Vectorization refers to how an attacker represents records for comparison before performing similarity evaluation.

\paragraph{Requirements.}
The following requirements ensure meaningful and comparable representations:
\begin{itemize}
    \item Numerical features are standardized, for example, using z-score normalization.
    \item Categorical features are encoded using one-hot, frequency-based, or hierarchical encodings.
    \item Representation consistency is maintained between the original dataset $\Dor$ and the released dataset $\Dpr$.
\end{itemize}

CVPL is agnostic to the specific encoding choice. The key requirement is semantic consistency across datasets: the same attribute must correspond to the same semantic dimension in the vector space for both $\Dor$ and $\Dpr$. The vectorization operator is assumed to be fixed and does not rely on access to ground-truth linkages or hidden training mechanisms.

\subsection{Latent Projection Operator}\label{sec:projection}

\paragraph{Definition.}
The projection operator maps feature vectors to a lower-dimensional latent space:
\begin{equation}\label{eq:projection-def}
    \oppsi(\mathbf{v}) = \mathbf{z} \in \Reals^k.
\end{equation}
Projection reflects the attacker's ability to exploit correlations and shared structure in the data representation.

\paragraph{Baseline choice.}
The default projection is Principal Component Analysis (PCA), fitted on the combined set $\{\opphi(x) : x \in \Dor \cup \Dpr\}$. Joint fitting is used by default, as it captures the latent structure shared between the original and released datasets, which enables linkage.

\paragraph{Projection fitting strategies.}

\begin{table}[htbp]
\centering
\caption{Projection fitting strategies and their usage scenarios.}
\label{tab:projection-strategies}
\begin{tabular}{@{}lll@{}}
\toprule
\textbf{Strategy} & \textbf{Description} & \textbf{Usage} \\
\midrule
Joint ($\Dor \cup \Dpr$) & Fit on pooled data & Default: captures shared latent structure \\
Source-only ($\Dor$) & Fit on original data & When $\Dpr$ exhibits distributional shift \\
Target-only ($\Dpr$) & Fit on protected data & Rare, loses source structure \\
Pre-fitted & External projection & Cross-run reproducibility \\
\bottomrule
\end{tabular}
\end{table}

Empirical ablation results (\cref{sec:results}) show limited sensitivity to this choice for CVPL-LR estimates.

\paragraph{Purpose.}
Latent projection serves three primary purposes:
\begin{enumerate}
    \item \textbf{Noise suppression}, by removing low-variance dimensions.
    \item \textbf{Extraction of dominant correlational structure}, preserved by utility-driven protection mechanisms.
    \item \textbf{Improved separability} between matching and non-matching candidate pairs.
\end{enumerate}

\paragraph{Theoretical justification.}
Among all linear projections of fixed dimension $k$, PCA maximizes retained variance. Under Gaussian or elliptical assumptions, this corresponds to preserving maximal mutual information between the original and projected representations~\cite{cover2006elements}. As a result, PCA is well-suited for revealing structural invariants that remain after privacy-preserving transformations.

\paragraph{Key insight.}
Protection mechanisms designed to preserve utility tend to preserve covariance structure. These structural invariants are captured by $\oppsi$ and constitute the primary source of residual linkability detected by CVPL.

\subsection{Similarity Function}\label{sec:similarity}

\paragraph{Definition.}
Similarity is computed in latent space using cosine similarity:
\begin{equation}\label{eq:cosine-similarity}
    s(\zx, \zy) = \frac{\langle \zx, \zy \rangle}{\|\zx\| \cdot \|\zy\|}.
\end{equation}
Cosine similarity is used as the baseline. Alternative similarity functions are evaluated in ablation studies.

\subsection{Linkage Predicate}\label{sec:linkage-predicate}

Given a threshold $\optau$, we define the linkage predicate:
\begin{equation}\label{eq:linkage-predicate}
    \ell_\optau(x, y) = \Ind\bigl[s(\oppsi(\opphi(x)), \oppsi(\opphi(y))) \geq \optau\bigr].
\end{equation}

\paragraph{Threshold interpretation.}
The threshold $\optau \in [0, 1]$ controls linkage strictness. Values near 0 treat any candidate as a potential link (high recall, low precision); values near 1 require near-identical embeddings (low recall, high precision). Intermediate values (\eg $\optau = 0.25$) represent moderate similarity requirements. The operational interpretation depends on the similarity function $s$; see \cref{sec:thresholds} for empirical calibration.

A record $x \in \Dor$ is considered \textbf{linkable} if:
\begin{equation}\label{eq:linkable-condition}
    \exists\, y \in \Dpr \text{ such that } \opB(x) = \opB(y) \land \ell_\optau(x, y) = 1.
\end{equation}
This definition corresponds to existential linkage and reflects a worst-case privacy risk perspective.

\subsection{Algorithm}\label{sec:algorithm}

\begin{algorithm}[htbp]
\caption{CVPL Risk Estimation}
\label{alg:cvpl-estimation}
\begin{algorithmic}[1]
\Require Original dataset $\Dor$, protected dataset $\Dpr$, blocking parameters $\lambda$, threshold $\optau$, similarity function $s(\cdot, \cdot)$
\Ensure Linkage rate $\CVPLLR \in [0, 1]$
\Statex
\State \textbf{// Phase 1: Blocking Construction}
\State Construct blocking operator $\opB$ from parameters $\lambda$
\For{each $x \in \Dor \cup \Dpr$}
    \State Compute $\opB(x)$
\EndFor
\Statex
\State \textbf{// Phase 2: Vectorization}
\For{each $x \in \Dor \cup \Dpr$}
    \State Compute $\vx \gets \opphi(x)$
\EndFor
\Statex
\State \textbf{// Phase 3: Projection}
\State Fit $\oppsi$ on $\{\vx : x \in \Dor \cup \Dpr\}$ \Comment{e.g., PCA}
\For{each $x \in \Dor \cup \Dpr$}
    \State Compute $\zx \gets \oppsi(\vx)$
\EndFor
\Statex
\State \textbf{// Phase 4: Linkage Assessment}
\State $\texttt{linked\_count} \gets 0$
\For{each $x \in \Dor$}
    \State $C(x) \gets \{y \in \Dpr : \opB(x) = \opB(y)\}$ \Comment{candidate set}
    \For{each $y \in C(x)$}
        \If{$s(\zx, \zy) \geq \optau$}
            \State $\texttt{linked\_count} \gets \texttt{linked\_count} + 1$
            \State \textbf{break} \Comment{existential linkage: one match suffices}
        \EndIf
    \EndFor
\EndFor
\Statex
\State \textbf{// Phase 5: Output}
\State \Return $\CVPLLR \gets \texttt{linked\_count} / |\Dor|$
\end{algorithmic}
\end{algorithm}

\subsection{Computational Complexity}\label{sec:complexity}

\begin{table}[htbp]
\centering
\caption{Computational complexity of CVPL phases.}
\label{tab:complexity}
\begin{tabular}{@{}ll@{}}
\toprule
\textbf{Phase} & \textbf{Complexity} \\
\midrule
Vectorization & $O((n + m) \cdot d)$ \\
Projection (PCA) & $O((n + m) \cdot d^2 + d^3)$ \\
Similarity evaluation & $O(n \cdot \bar{B} \cdot k)$ \\
\textbf{Total} & $O((n + m) \cdot d^2 + n \cdot \bar{B} \cdot k)$ \\
\bottomrule
\end{tabular}
\end{table}

Here $\bar{B}$ denotes the average block size. Effective blocking ensures $\bar{B} \ll m$.

\subsection{Scalability Considerations}\label{sec:scalability}

For large-scale deployments ($n > 10^6$), direct evaluation of all candidate similarities becomes computationally impractical. CVPL addresses scalability through a combination of structural constraints and approximate evaluation strategies that preserve the attacker model while reducing computational cost.

\paragraph{Block size management.}
Blocking defines the primary scalability lever by limiting the candidate search space. However, real-world quasi-identifiers often induce heavy-tailed block size distributions. To address this, several complementary strategies can be applied.

\begin{table}[htbp]
\centering
\caption{Block size management strategies.}
\label{tab:block-management}
\begin{tabular}{@{}lll@{}}
\toprule
\textbf{Strategy} & \textbf{Description} & \textbf{Use Case} \\
\midrule
Cutoff & Limit block size & Fast, reduced recall \\
Sub-blocking & Secondary partitioning & Preferred \\
ANN & Approximate NN within blocks & Very large blocks \\
Sampling & Random subset & Probabilistic estimates \\
\bottomrule
\end{tabular}
\end{table}

\textbf{Note:} The heuristics in \cref{tab:block-management} may affect the monotonicity guarantee of \cref{thm:monotonicity}, which assumes exhaustive candidate evaluation. When approximations are used, CVPL-LR should be interpreted as a lower bound under the specified evaluation budget.

A practical deployment strategy includes monitoring block size distributions, applying sub-blocking above the 95th percentile, and using ANN methods such as FAISS~\cite{johnson2019billion} for blocks exceeding 10,000 records.

\paragraph{Performance optimizations.}
In addition to structural constraints, several implementation-level optimizations substantially reduce computational overhead without altering CVPL semantics:
\begin{itemize}
    \item \textbf{Chunked top-1 similarity evaluation.} Similarities are computed in fixed-size chunks (\eg 1024 vectors), avoiding materialization of the full similarity matrix. This reduces memory complexity from $O(N^2)$ to $O(N \cdot \text{chunk})$.
    
    \item \textbf{Constant-time index lookup.} Hash-based local index mapping replaces linear searches within blocks, reducing candidate lookup from $O(|B|)$ to $O(1)$.
    
    \item \textbf{Sampled similarity distribution.} For non-matching pairs, similarity distributions are estimated using fixed-size samples (\eg 100 samples per record), preserving distributional shape while avoiding quadratic storage.
    
    \item \textbf{Threshold reuse.} Similarities are computed once and re-thresholded across multiple $\optau$ values, enabling efficient computation of risk surfaces and yielding up to $30\times$ speedup compared to recomputation.
\end{itemize}

\paragraph{Empirical scaling.}
On commodity hardware (32 GB RAM, 8 CPU cores), the following performance improvements were observed:

\begin{table}[htbp]
\centering
\caption{Performance improvements with optimizations.}
\label{tab:performance}
\begin{tabular}{@{}rrrr@{}}
\toprule
\textbf{Records} & \textbf{Time (baseline)} & \textbf{Time (optimized)} & \textbf{Speedup} \\
\midrule
1,000 & ${\sim}110$\,s & ${\sim}11$\,s & $10\times$ \\
5,000 & ${\sim}550$\,s & ${\sim}51$\,s & $10\times$ \\
10,000 & ${\sim}1{,}100$\,s & ${\sim}229$\,s & ${\sim}5\times$ \\
\bottomrule
\end{tabular}
\end{table}

These results demonstrate that CVPL scales to millions of records under realistic deployment conditions while preserving the intended threat model and linkage semantics.

\subsection{Relationship to Fellegi--Sunter}\label{sec:fellegi-sunter}

\begin{proposition}[Fellegi--Sunter as a Special Case of CVPL]\label{prop:fs-special-case}
Fellegi--Sunter (FS) probabilistic record linkage~\cite{fellegi1969theory} can be expressed as a special case of CVPL under the following restrictive assumptions:
\begin{enumerate}
    \item All features are binary or categorical.
    \item The projection operator is the identity, $\oppsi = I$.
    \item Feature independence is assumed.
    \item Similarity is defined as a log-likelihood ratio.
\end{enumerate}
\end{proposition}

Under these conditions, the CVPL similarity score reduces to the classical Fellegi--Sunter composite weight.

\begin{proof}[Proof sketch]
Let $\gamma \in \{0, 1\}^d$ denote the binary agreement vector between records $x$ and $y$, where $\gamma_i = 1$ indicates agreement on feature $i$. The Fellegi--Sunter likelihood ratio is given by:
\begin{equation}\label{eq:fs-likelihood}
    \log R(\gamma) = \sum_i \gamma_i \log\frac{m_i}{u_i} + \sum_i (1 - \gamma_i) \log\frac{1 - m_i}{1 - u_i} = \sum_i w_i \gamma_i + c,
\end{equation}
where $m_i$ and $u_i$ denote match and non-match agreement probabilities, respectively.

With identity projection ($\oppsi = I$), binary feature encoding, and a weighted dot-product similarity function, the CVPL score coincides with the FS composite weight.
\end{proof}

CVPL generalizes this framework by supporting continuous features, capturing dependencies through latent projection, and avoiding independence assumptions.

\begin{proposition}[Systematic Bias of Fellegi--Sunter under Representation Shift]\label{prop:fs-bias}
When protected data introduces generalization, suppression, or perturbation on a subset of attributes, the Fellegi--Sunter linkage score becomes systematically biased due to residual correlations in non-generalized attributes. This leads to consistent over-linking.
\end{proposition}

\paragraph{Interpretation.}
Fellegi--Sunter is calibrated under the assumption that observed agreement patterns arise from true matches. In protected datasets, however, utility-preserving mechanisms often reduce agreement on selected attributes (\eg via generalization) while preserving correlated structure in others. As a result, FS overestimates match likelihood based on the remaining agreements, despite the absence of identity correspondence.

CVPL avoids this failure mode by operating in a latent space that captures protection-induced structural invariants directly, rather than relying on per-attribute agreement counts.

\paragraph{Empirical confirmation.}
At a scale of 10,000 records, we observe:
\begin{itemize}
    \item \textbf{Fellegi--Sunter:} FS-LR = 93\%, Precision@1 = 13\%
    \item \textbf{CVPL:} CVPL-LR = 26\%, Precision@1 = 34\%
\end{itemize}
These results confirm that classical FS linkage substantially over-links under representation shift, while CVPL provides a more reliable estimate of residual linkage risk (see \cref{sec:results}).

\subsection{Progressive Blocking and Monotonicity}\label{sec:progressive-blocking}

\begin{definition}[Blocking Relaxation]\label{def:blocking-relaxation}
A blocking scheme $B_2$ is a \emph{relaxation} of $B_1$ (denoted $B_1 \preceq B_2$) if it is no more restrictive than $B_1$, \ie every pair that matches under $B_1$ also matches under $B_2$:
\begin{equation}\label{eq:relaxation-def}
    \forall x, y\colon B_1(x) = B_1(y) \Rightarrow B_2(x) = B_2(y).
\end{equation}
\end{definition}

Equivalently, $B_2$ can be obtained by \textbf{dropping constraints} or \textbf{merging blocks} of $B_1$. Informally, relaxation increases candidate sets and therefore can only increase (or preserve) existential linkability.

For any record $x \in \Dor$, define the candidate set under blocking $B$:
\begin{equation}\label{eq:candidate-under-B}
    \Bcand{x} = \{y \in \Dpr : B(x) = B(y)\}.
\end{equation}
Then $B_1 \preceq B_2 \Rightarrow C_{B_1}(x) \subseteq C_{B_2}(x)$ for all $x$.

\begin{theorem}[Monotonicity under Relaxation]\label{thm:monotonicity}
Let $B_1 \preceq B_2$ be blocking relaxation (\ie $B_2$ relaxes $B_1$). Then, for all thresholds $\optau$:
\begin{equation}\label{eq:monotonicity}
    \CVPLLR(B_1, \optau) \leq \CVPLLR(B_2, \optau).
\end{equation}
\end{theorem}

\begin{proof}
Fix any $x \in \Dor$. Let $C_1(x) = C_{B_1}(x)$ and $C_2(x) = C_{B_2}(x)$. By relaxation, $C_1(x) \subseteq C_2(x)$.

Define the existential linkage indicator under $B_i$ as:
\begin{equation}
    L_i(x) = \Ind[\exists\, y \in C_i(x) : s(x, y) \geq \optau].
\end{equation}
Since $C_1(x) \subseteq C_2(x)$, if $L_1(x) = 1$, then the same witness $y$ belongs to $C_2(x)$, hence $L_2(x) = 1$. Therefore $L_1(x) \leq L_2(x)$ pointwise for all $x$.

Averaging over $\Dor$:
\begin{equation}
    \CVPLLR(B_i, \optau) = \frac{1}{|\Dor|} \sum_{x \in \Dor} L_i(x),
\end{equation}
so $\sum_x L_1(x) \leq \sum_x L_2(x)$, and thus $\CVPLLR(B_1, \optau) \leq \CVPLLR(B_2, \optau)$.
\end{proof}

\textbf{Note:} \Cref{thm:monotonicity} assumes exhaustive evaluation over candidate sets $C(x)$ under nested blocking refinements. When scalability heuristics are used (\eg truncation, sampling, or ANN search), strict monotonicity may not hold; however, early-stage estimates remain empirically stable in our experiments and can be treated as conservative lower bounds under the chosen evaluation budget.

\begin{corollary}[Anytime Lower Bound]\label{cor:anytime}
Let $B_1 \preceq B_2 \preceq \cdots \preceq B_k$ be a progressive relaxation ladder (coarse-to-broad matching). Then the sequence $R_i = \CVPLLR(B_i, \optau)$ is non-decreasing. Stopping early yields a valid \textbf{lower bound} on the final risk estimate:
\begin{equation}
    R_1 \leq R_2 \leq \cdots \leq R_k, \quad\text{and}\quad R_i \leq R_k \;\;\forall\, i < k.
\end{equation}
\end{corollary}

\paragraph{Interpretation.}
Progressive relaxation provides a controllable ``attacker strength'' ladder: each step broadens candidate sets (more comparisons, more opportunities to find a match), so existential linkage risk can only increase.

\paragraph{Practical implications.}
Progressive relaxation enables:
\begin{enumerate}
    \item \textbf{Fast screening:} start with a restrictive blocking (small candidate sets) to obtain a quick \textbf{lower bound}.
    \item \textbf{Attacker-strength calibration:} interpret each relaxation step as a stronger adversary (more linking power).
    \item \textbf{Resource budgeting:} trade computation time for tightness of the estimate (approach the final $R_k$).
    \item \textbf{Convergence monitoring:} track the incremental gain $\Delta_i = R_i - R_{i-1}$ and stop when it becomes negligible.
\end{enumerate}

\begin{algorithm}[htbp]
\caption{Progressive CVPL Evaluation}
\label{alg:progressive-cvpl}
\begin{algorithmic}[1]
\Require $\Dor$, $\Dpr$, relaxation ladder $B_1 \preceq B_2 \preceq \cdots \preceq B_k$, threshold $\optau$, convergence tolerance $\varepsilon$
\Ensure CVPL-LR sequence with convergence-from-below guarantee
\Statex
\State $R_{\text{prev}} \gets 0$
\For{$i = 1$ \textbf{to} $k$}
    \State $R_i \gets \CVPLLR(\Dor, \Dpr, B_i, \optau)$
    \State $\Delta \gets R_i - R_{\text{prev}}$
    \If{$\Delta < \varepsilon$}
        \State \Return $R_i$ \Comment{converged}
    \EndIf
    \State $R_{\text{prev}} \gets R_i$
\EndFor
\State \Return $R_k$
\end{algorithmic}
\end{algorithm}


\section{Evaluation Protocol}\label{sec:protocol}

\subsection{Ground Truth}\label{sec:ground-truth}

Let $\mathrm{id}(x)$ identify a hidden identifier associated with the record $x$. The true linkage relation is defined as:
\begin{equation}\label{eq:true-linkage}
    \Lstar = \{(x, y) \in \Dor \times \Dpr : \mathrm{id}(x) = \mathrm{id}(y)\}.
\end{equation}
Ground truth is used exclusively for evaluation purposes. The CVPL algorithm itself never accesses hidden identifiers and operates solely on the released datasets $\Dor$ and $\Dpr$.

\subsection{Ground Truth Availability Scenarios}\label{sec:gt-scenarios}

\begin{table}[htbp]
\centering
\caption{Ground truth availability scenarios and their implications for CVPL evaluation.}
\label{tab:gt-scenarios}
\begin{tabular}{@{}llll@{}}
\toprule
\textbf{Scenario} & \textbf{Ground Truth} & \textbf{CVPL Executable} & \textbf{Validation Possible} \\
\midrule
Simulation & Full & Yes & Yes \\
Controlled release & Full & Yes & Yes \\
Third-party audit & Partial or none & Yes & Limited \\
Blind audit & None & Yes & No \\
\bottomrule
\end{tabular}
\end{table}

\paragraph{Implication.}
CVPL can be executed without ground truth in all scenarios. However, quantitative validation of linkage risk estimates requires some form of ground truth. In blind audit scenarios, CVPL provides risk estimates that cannot be directly verified but remain meaningful for relative comparison between protection mechanisms.

\subsection{Self-Linkage as Diagnostic}\label{sec:self-linkage}

When ground truth is unavailable, \textbf{self-linkage} can be used as a diagnostic tool. Let $\Dpr$ be obtained from $\Dor$ through controlled perturbation. The self-linkage rate is defined as:
\begin{equation}\label{eq:self-linkage}
    \text{Self-LR} = \CVPLLR(\Dor, \Dpr).
\end{equation}
A high self-linkage rate (approaching 1) indicates sufficient discriminative power of the feature representation and similarity pipeline. Conversely, low self-linkage suggests that the chosen features, encodings, or projections are insufficiently informative for linkage assessment.

\subsection{CVPL-Linkage Rate}\label{sec:cvpl-lr-metric}

\paragraph{Primary metric.}
For a fixed similarity threshold $\optau$, the CVPL linkage rate is:
\begin{equation}\label{eq:cvpl-lr-metric}
    \CVPLLR(\optau) = \frac{1}{|\Dor|} \sum_{x \in \Dor} \Ind\bigl[\exists\, y \in \mathcal{C}(x) : s(x, y) \geq \optau\bigr],
\end{equation}
where $\mathcal{C}(x) = \{y \in \Dpr : \opB(x) = \opB(y)\}$ is the candidate set induced by blocking. The score $s(x, y)$ denotes the similarity computed after the full CVPL pipeline, \ie on the projected representations $\oppsi(\opphi(x))$ and $\oppsi(\opphi(y))$.

\paragraph{Interpretation.}
$\CVPLLR(\optau)$ is the fraction of original records that admit at least one \emph{plausible} link in the protected dataset under the assumed attacker model. It measures \textbf{linkage feasibility} (existential linkability): that is, the risk that an attacker could plausibly claim a link, rather than de-anonymization accuracy.

\paragraph{Semantic clarification.}
$\CVPLLR$ induces a threshold-indexed family of set functions over subsets $A \subseteq \Dor$:
\begin{equation}\label{eq:set-function}
    \mu_\optau(A) = \frac{|\{x \in A : \exists\, y \in \Dpr,\, s(x, y) \geq \optau\}|}{|\Dor|}, \quad \mu_\optau(\Dor) = \CVPLLR(\optau).
\end{equation}
This enables ordering and comparison of protection mechanisms across attacker strictness levels.

\subsection{Auxiliary Metrics}\label{sec:auxiliary-metrics}

CVPL-LR measures \textbf{linkage feasibility}. When ground truth is available for evaluation purposes only, we additionally report auxiliary metrics that separate correct from spurious linkability and support threshold calibration and method comparison.

\paragraph{True Link Rate (TLR).}
The probability that the true counterpart exceeds the similarity threshold, \emph{conditional on not being eliminated by blocking}:
\begin{equation}\label{eq:tlr}
    \TLR(\optau) = \Prob\bigl[s(x, \ystar) \geq \optau \mid (x, \ystar) \in \Lstar,\, \opB(x) = \opB(\ystar)\bigr].
\end{equation}
TLR measures the sensitivity of the similarity stage within admissible candidate sets and isolates the effect of similarity scoring from blocking-induced constraints.

\paragraph{False Link Rate (FLR).}
The probability that at least one \emph{incorrect} candidate exceeds the threshold within the same block:
\begin{equation}\label{eq:flr}
    \FLR(\optau) = \Prob\bigl[\exists\, y \neq \ystar : s(x, y) \geq \optau \mid \opB(x) = \opB(y)\bigr].
\end{equation}
FLR captures the propensity of the linkage mechanism to produce \emph{spurious but plausible} links, which directly contributes to existential linkage risk.

\paragraph{Top-1 Precision (Precision@1).}
The probability that the highest-similarity candidate corresponds to the true match:
\begin{equation}\label{eq:prec-at-1}
    \PrecOne = \Prob\bigl[\argmax_{y \in \Dpr} s(x, y) = \ystar \mid (x, \ystar) \in \Lstar\bigr].
\end{equation}
Precision@1 measures \textbf{de-anonymization accuracy} rather than feasibility and is reported separately from CVPL-LR.

\paragraph{Note.}
TLR, FLR, and Precision@1 are \textbf{evaluation metrics} that require access to the true linkage relation $\Lstar$. They are \textbf{not observable} to the attacker and are \textbf{not assumed} under the CVPL threat model.

\subsection{Recall Decomposition}\label{sec:recall-decomposition}

To attribute linkage success (or failure) transparently, we decompose total recall into a \textbf{structural component} induced by blocking and a \textbf{matching component} induced by similarity-based decisions. This separation enables honest attribution of errors to threat-model constraints rather than to the similarity mechanism itself.

\paragraph{Blocking recall.}
The fraction of true linked pairs that are \emph{not eliminated} by the blocking constraint:
\begin{equation}\label{eq:blocking-recall}
    \Rblock = \Prob\bigl[\opB(x) = \opB(\ystar) \mid (x, \ystar) \in \Lstar\bigr].
\end{equation}
Blocking recall captures the \emph{structural ceiling} on achievable linkage under the assumed attacker model: no similarity-based method can recover links excluded by blocking.

\paragraph{Within-block matching recall.}
The probability that the true pair exceeds the similarity threshold, \emph{conditional on being in the same block}:
\begin{equation}\label{eq:matching-recall}
    R_{\text{match}} = \Prob\bigl[s(x, \ystar) \geq \optau \mid (x, \ystar) \in \Lstar,\, \opB(x) = \opB(\ystar)\bigr].
\end{equation}
This term isolates the effectiveness of the similarity and projection stages independently of blocking constraints.

\paragraph{Total recall.}
The overall recall decomposes multiplicatively as:
\begin{equation}\label{eq:total-recall}
    R_{\text{total}} = \Rblock \cdot R_{\text{match}}.
\end{equation}

\paragraph{Interpretation.}
This decomposition separates \textbf{structural constraints} imposed by the threat model (blocking) from \textbf{similarity-based detection} (matching). It prevents misattribution of linkage failure to the similarity mechanism when it is, in fact, induced by coarse blocking or conservative attacker assumptions.

Low total recall may arise from intentionally restrictive blocking, even when similarity scoring is highly discriminative, highlighting the importance of analyzing both components jointly.

\subsection{Baseline Methods}\label{sec:baselines}

We compare CVPL against established and diagnostic baseline methods. Unless stated otherwise, all baselines are evaluated under the same blocking and candidate-generation constraints to ensure comparability of linkage feasibility estimates.

\paragraph{B1. Fellegi--Sunter (FS).}
Classical probabilistic record linkage is performed within blocks, based on attribute agreement patterns and log-likelihood ratio weighting under conditional independence assumptions~\cite{fellegi1969theory}. This baseline represents the canonical linkage model and serves as a reference point for comparison with projection-based methods.

\paragraph{B2. Random-within-Block.}
For each $x \in \Dor$, a candidate $y \in \mathcal{C}(x)$ is selected uniformly at random. This baseline provides a trivial lower bound and a sanity check on linkage rates, corresponding to a non-informative attacker restricted only by blocking constraints.

\paragraph{B3. No-Projection Ablation.}
CVPL with identity projection $\oppsi = I$, isolating the contribution of latent projection to linkage feasibility. This ablation distinguishes gains due to geometric structure captured by projection from those attributable solely to feature-level similarity.

\paragraph{B4. Distance to Closest Record (DCR).}
Minimum Euclidean distance between each protected record and the closest original record. DCR is commonly used in synthetic data evaluation as a heuristic measure of disclosure risk, despite lacking an explicit attacker or decision model~\cite{zhao2021ctab}.

\paragraph{B5. Nearest Neighbor Distance Ratio (NNDR).}
Ratio of distances to the first and second nearest neighbors in the original dataset~\cite{christen2012data}. NNDR is used as a heuristic proxy for match uniqueness and local ambiguity in synthetic data evaluation but does not directly model linkage decisions.

\paragraph{Note.}
Baselines B4 and B5 are included for diagnostic comparison with prior synthetic data literature. Unlike CVPL and FS, they do not correspond to explicit linkage decision rules and therefore should be interpreted as heuristic risk indicators rather than operational attacker models.


\section{Threshold Calibration and Risk Surfaces}\label{sec:thresholds}

\subsection{Similarity Distributions}\label{sec:similarity-distributions}

To analyze linkage behavior, we consider the empirical distributions of similarity scores under the ground truth.

\paragraph{True-match similarity.}
\begin{equation}\label{eq:true-match-dist}
    \Splus = \{s(x, \ystar) : (x, \ystar) \in \Lstar\}.
\end{equation}

\paragraph{Non-match similarity.}
\begin{equation}\label{eq:non-match-dist}
    \Sminus = \{s(x, y) : y \neq \ystar,\, \opB(x) = \opB(y)\}.
\end{equation}

The degree of overlap between $\Splus$ and $\Sminus$ determines the discriminability of the linkage process and directly informs threshold selection. Strong overlap indicates that plausible false links coexist with true matches at similar similarity levels.

\subsection{Threshold as an Operational Parameter}\label{sec:threshold-operational}

The similarity threshold $\optau$ is not treated as a model hyperparameter to be optimized. Instead, it is an \textbf{operational parameter} reflecting the attacker's strictness.

Different attackers may tolerate different false-positive rates depending on context, incentives, and risk tolerance. Consequently, linkage risk should be evaluated across a range of thresholds rather than reported at a single, arbitrarily chosen value.

\subsection{Threshold Selection Strategies}\label{sec:threshold-strategies}

We consider several principled strategies for interpreting risk across thresholds.

\paragraph{Strategy A: Precision-constrained threshold.}
\begin{equation}\label{eq:precision-constrained}
    \optau^* = \min\{\optau : \FLR(\optau) \leq \alpha\},
\end{equation}
where $\alpha$ is the target false-link rate. This strategy yields a conservative estimate that controls spurious link claims.

\paragraph{Strategy B: Worst-case risk.}
\begin{equation}\label{eq:worst-case-risk}
    R_{\max} = \sup_{\optau} \CVPLLR(\optau).
\end{equation}
This represents an upper bound on linkage feasibility across reasonable attacker strictness levels.

\paragraph{Strategy C: Integrated risk.}
\begin{equation}\label{eq:integrated-risk}
    R_{\text{int}} = \int_{\optau_{\min}}^{\optau_{\max}} \CVPLLR(\optau)\, d\optau.
\end{equation}
This aggregated measure reduces dependence on a single threshold choice and captures average risk over an operational range.

\paragraph{Default reporting recommendation.}
We recommend reporting the full risk curve over $\optau \in [0.80, 0.99]$, with $\optau = 0.90$ used for concise, single-number summaries when required.

\subsection{Risk Surfaces}\label{sec:risk-surfaces}

\paragraph{Definition.}
Let $\lambda$ denote protection parameters controlling the strength or configuration of the data transformation. The CVPL risk surface is defined as:
\begin{equation}\label{eq:risk-surface}
    \RiskSurf{\lambda}{\optau} = \CVPLLR(\lambda, \optau),
\end{equation}
where $\optau$ is the similarity threshold representing attacker strictness.

\paragraph{Interpretation.}
Linkage risk depends jointly on two orthogonal dimensions:
\begin{itemize}
    \item \textbf{Protection strength ($\lambda$)}, which determines how aggressively the data is transformed, and
    \item \textbf{Attacker strictness ($\optau$)}, which determines how confident a linkage claim must be to be considered plausible.
\end{itemize}

\paragraph{Analytical role.}
Risk surfaces enable the following analyses:
\begin{itemize}
    \item identification of regions of \emph{stable low risk} versus regions of \emph{high sensitivity} to parameter changes,
    \item assessment of whether formal compliance thresholds correspond to empirically low linkage risk,
    \item analysis of \emph{utility--risk trade-offs} across protection regimes and attacker models.
\end{itemize}

By varying $\lambda$ and $\optau$, CVPL produces a continuous, interpretable landscape of residual linkage risk rather than a single point estimate.

\subsection{Key Observation}\label{sec:threshold-observation}

Empirical results show that formal compliance (\eg satisfying $k$-anonymity) does not imply uniformly low linkage risk across attacker strictness levels. Even when a dataset meets the required $k$, the risk surface $\RiskSurf{\lambda}{\optau}$ may remain non-negligible for a broad range of realistic thresholds $\optau$.

Across protection regimes, risk typically decreases \textbf{gradually} as protection strength increases, rather than collapsing at a single compliance boundary. This supports treating privacy as a \textbf{continuous, scenario-dependent quantity}, and motivates evaluating protection mechanisms via risk surfaces instead of single-threshold pass/fail criteria.


\section{Experimental Design}\label{sec:experiments}

\subsection{Simulation Scenario: Marketing Touchpoint}\label{sec:simulation-scenario}

We construct a controlled, reproducible simulation inspired by marketing analytics, where the analytic goal is to associate advertising exposure with subsequent purchase behavior. The setting reflects common data-sharing workflows in which privacy protection is applied, yet downstream utility (segmentation, attribution signals, timing effects) must be retained.

\paragraph{Dataset schema.}
The simulated dataset follows the schema in \cref{tab:dataset-schema}.

\begin{table}[htbp]
\centering
\caption{Simulated dataset schema for marketing touchpoint scenario.}
\label{tab:dataset-schema}
\begin{tabular}{@{}lll@{}}
\toprule
\textbf{Field} & \textbf{Type} & \textbf{Description} \\
\midrule
\texttt{person\_id} & Hidden & Ground-truth identifier (evaluation only) \\
\texttt{age} & Numeric & Age in years \\
\texttt{gender} & Categorical & Gender (M/F/Other) \\
\texttt{region} & Categorical & Hierarchical geographic region \\
\texttt{brand\_product} & Categorical & Product category or brand \\
\texttt{purchase\_place} & Categorical & Type of purchase location \\
\texttt{purchase\_time} & Timestamp & Purchase datetime \\
\texttt{days\_after\_ad} & Numeric & Delay between ad exposure and purchase \\
\texttt{ad\_channel} & Categorical & Advertising channel (Search, Social, Video, etc.) \\
\bottomrule
\end{tabular}
\end{table}

The hidden identifier is never used by CVPL and serves exclusively for evaluation.

\paragraph{Simulation realism.}
While synthetic, the simulation is not random. It is constructed to preserve several real-world properties that are relevant for linkage feasibility:
\begin{itemize}
    \item \textbf{Correlated quasi-identifiers:} age, region, and gender follow realistic joint distributions.
    \item \textbf{Behavioral fingerprints:} channel preference and purchase delay are conditioned on demographics.
    \item \textbf{Heavy-tailed categorical frequencies:} products and channels follow Zipf-like popularity profiles.
    \item \textbf{Temporal structure:} purchase activity exhibits day-of-week and hour-of-day effects.
\end{itemize}

CVPL does not depend on the scenario being ``easy.'' Rather, it probes whether \emph{utility-preserving releases} retain structural invariants---correlations, latent clusters, and temporal regularities---that can make records plausibly linkable even under formal compliance.

\subsection{Data Generating Process}\label{sec:data-generation}

The simulator is designed to capture realistic dependencies commonly observed in marketing data. The generative process proceeds as follows:
\begin{enumerate}
    \item \textbf{Demographics:} $(\text{age}, \text{gender}, \text{region}) \sim P_{\text{pop}}$
    \item \textbf{Advertising channel:} $\text{channel} \sim P(\text{channel} \mid \text{region}, \text{age}, \text{gender})$
    \item \textbf{Response lag:} $\text{days\_after\_ad} \sim P(\text{lag} \mid \text{channel}, \text{age})$
    \item \textbf{Purchase choice:} $\text{product} \sim P(\text{product} \mid \text{channel}, \text{demographics})$
    \item \textbf{Purchase time:} $\text{time} = t_0 + \text{lag} + \epsilon_t$, with calendar and diurnal effects.
\end{enumerate}

\paragraph{Design principle.}
These conditional dependencies represent structural invariants that are typically preserved by utility-oriented protection mechanisms. CVPL is specifically designed to detect residual linkability arising from such preserved structure.

\subsection{Controlled Anomalies}\label{sec:controlled-anomalies}

To stress-test linkage risk under realistic conditions, we introduce controlled anomalies into the data:
\begin{itemize}
    \item \textbf{Rare channels:} Low-frequency, high-specificity advertising channels.
    \item \textbf{Long-tail products:} Rare product categories and combinations.
    \item \textbf{Temporal anomalies:} Purchases at unusual hours and holiday spikes.
    \item \textbf{Lag outliers:} Near-instant responses ($\Delta t \approx 0$) and long-delayed purchases.
\end{itemize}

Anomaly prevalence is controlled by the parameter $p_{\text{out}} \in \{0, 0.01, 0.05\}$.

\subsection{Protection Mechanisms}\label{sec:protection-mechanisms}

We evaluate three families of protection mechanisms, covering common approaches used in practice. Each family is parameterized to allow controlled variation of protection strength.

\paragraph{Family A: k-Anonymity--Based Generalization.}
This family applies structural generalization to quasi-identifiers~\cite{sweeney2002}: age binning, region coarsening (hierarchical generalization), and time rounding. The anonymity parameter is varied as $k \in \{2, 3, 5, 7, 10, 15, 20, 30, 50\}$.

\paragraph{Family B: Perturbation-Based Protection.}
This family applies noise and randomization while preserving overall schema~\cite{christen2012data}: time jitter, noise added to response lag, and categorical value swapping. Protection strength is controlled via discrete levels:
\begin{equation}\label{eq:perturbation-levels}
    \text{level} \in \{\text{low}, \text{medium}, \text{high}\}.
\end{equation}
Note that perturbation inherently distorts marginal and joint distributions; this family represents a common but utility-degrading approach, included to evaluate CVPL across a spectrum of protection strategies.

\paragraph{Family C: Utility-Preserving Synthetic Data.}
This family evaluates synthetic data generation methods designed to preserve analytical utility~\cite{xu2019modeling}: preservation of marginal distributions and preservation of selected inter-attribute dependencies. The strength of protection is controlled through correlation retention:
\begin{equation}\label{eq:synthetic-rho}
    \rho \in \{0.95, 0.90, 0.85, 0.80, 0.75, 0.70, 0.60\}.
\end{equation}
Higher values of $\rho$ correspond to stronger utility preservation and weaker protection, while lower values increase privacy at the cost of structural fidelity.

\subsection{CVPL Configuration and Evaluation Setup}\label{sec:cvpl-config}

Unless stated otherwise, all experiments use a fixed CVPL configuration to ensure comparability across protection mechanisms.

\paragraph{Blocking.}
Candidate generation is constrained by quasi-identifier blocking using: \texttt{age\_bin} (generalized age interval), \texttt{region\_level} (hierarchical region level), and \texttt{gender}. Thus, for each record $x \in \Dor$, the candidate set is:
\begin{equation}\label{eq:candidate-set-exp}
    \mathcal{C}(x) = \{y \in \Dpr : \opB(x) = \opB(y)\}.
\end{equation}

\paragraph{Vectorization.}
Records are mapped to numeric feature vectors using standardized numeric features (z-score normalization) and one-hot encoding for categorical features, with identical encoding schemes applied to $\Dor$ and $\Dpr$.

\paragraph{Projection.}
We apply PCA as the default latent projection, fitted on pooled data $\Dor \cup \Dpr$, retaining 90\% of explained variance.

\paragraph{Similarity.}
Similarity is computed in latent space using cosine similarity:
\begin{equation}\label{eq:similarity-exp}
    s(x, y) = \simcos(\oppsi(\opphi(x)), \oppsi(\opphi(y))).
\end{equation}

\paragraph{Threshold sweep.}
Risk is evaluated over a threshold range $\optau \in [0.70, 0.99]$, capturing attacker strictness from permissive matching to near-exact similarity requirements.

This configuration reflects a realistic and conservative attacker model, balancing computational tractability with threat-model plausibility.

\subsection{Experimental Grid}\label{sec:experimental-grid}

Experiments are conducted over a factorial grid defined by the factors in \cref{tab:experimental-grid}.

\begin{table}[htbp]
\centering
\caption{Experimental factorial grid.}
\label{tab:experimental-grid}
\begin{tabular}{@{}ll@{}}
\toprule
\textbf{Factor} & \textbf{Levels} \\
\midrule
Protection family & A, B, C \\
Protection strength & 5 levels per family \\
Outlier fraction & 0, 0.01, 0.05 \\
Random seeds & 5 per configuration \\
\bottomrule
\end{tabular}
\end{table}

\paragraph{Protection-strength levels.}
Within each family, we evaluate:
\begin{itemize}
    \item \textbf{Family A (k-anonymity):} $k \in \{2, 3, 5, 7, 10, 15, 20, 30, 50\}$ (9 levels);
    \item \textbf{Family B (perturbation):} $\{\text{low}, \text{medium}, \text{high}\}$ (3 levels);
    \item \textbf{Family C (synthetic):} $\rho \in \{0.95, 0.90, 0.85, 0.80, 0.75, 0.70, 0.60\}$ (7 levels),
\end{itemize}
for a total of \textbf{19 protection configurations} across families.

All experiments use $n = 10{,}000$ records. Results are reported as averages across random seeds, with variability shown where relevant (\eg confidence intervals for key metrics).

\subsection{Utility Metrics}\label{sec:utility-metrics}

We measure utility using multiple proxy metrics to avoid single-metric bias. Each metric captures a different aspect of preservation between the original dataset $\Dor$ and the protected dataset $\Dpr$.

\begin{table}[htbp]
\centering
\caption{Utility metrics for evaluating protection mechanisms.}
\label{tab:utility-metrics}
\begin{tabular}{@{}lll@{}}
\toprule
\textbf{Metric} & \textbf{Definition} & \textbf{Interpretation} \\
\midrule
Correlation preservation & $\rho(\mathrm{Corr}(\Dor), \mathrm{Corr}(\Dpr))$ & Structural fidelity \\
Marginal similarity & $1 - \mathrm{JSD}(p_{\Dor}, p_{\Dpr})$ & Distribution preservation \\
Downstream accuracy & $\mathrm{Acc}_{\Dpr} / \mathrm{Acc}_{\Dor}$ & Predictive utility retained \\
Pairwise distance preservation & $\rho(d_{\Dor}, d_{\Dpr})$ & Geometric fidelity \\
\bottomrule
\end{tabular}
\end{table}

\paragraph{Note.}
These are proxy measures; we do not claim to define a universal utility metric. Instead, we use a diversified set to reduce bias from any single indicator. Our empirical focus is on the observed tendency that mechanisms preserving more structure (higher utility under the above proxies) also tend to preserve more linkability (higher CVPL risk).


\section{Experimental Results}\label{sec:results}

This section presents an empirical evaluation of CVPL across multiple protection mechanisms, parameter regimes, and comparison baselines. All experiments are conducted on datasets of $n = 10{,}000$ records, averaged over 5 random seeds, with $\optau = 0.9$ as the primary threshold unless stated otherwise.

\subsection{CVPL-LR vs.\ Protection Strength (k-Anonymity)}\label{sec:results-kanon}

\begin{figure}[htbp]
    \centering
    \includegraphics[width=0.9\textwidth]{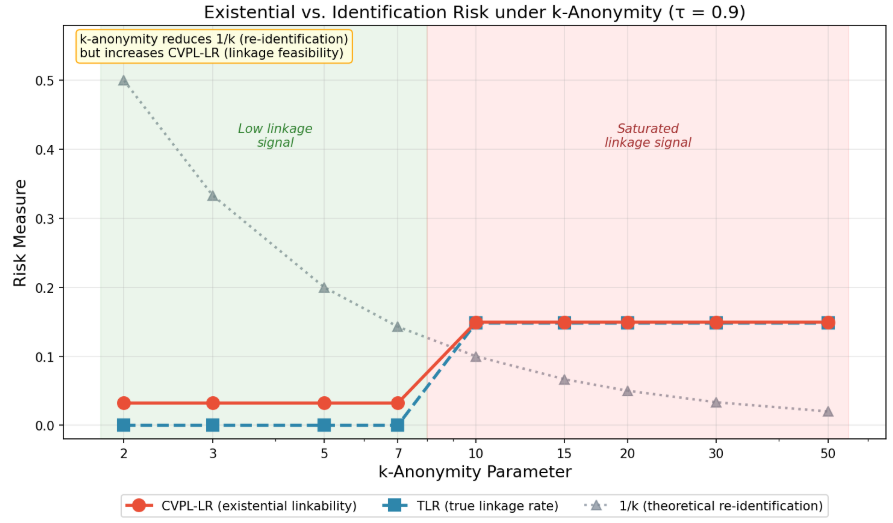}
    \caption{Existential linkage risk (CVPL-LR) versus identification risk ($1/k$) under k-anonymity protection.}
    \label{fig:g1-existential-vs-identification}
\end{figure}

\paragraph{Finding R1.}
CVPL-LR and traditional re-identification risk ($1/k$) exhibit \textbf{opposite trends} with respect to the k-anonymity parameter:

\begin{table}[htbp]
\centering
\caption{CVPL-LR versus traditional risk metrics under k-anonymity.}
\label{tab:cvpl-vs-risk}
\begin{tabular}{@{}llrrr@{}}
\toprule
\textbf{Regime} & \textbf{k} & \textbf{CVPL-LR} & \textbf{TLR} & \textbf{1/k} \\
\midrule
Low $k$ & 2--7 & 0.032 (3.2\%) & 0.00 (0\%) & 0.14--0.50 \\
High $k$ & 10--50 & 0.149 (14.9\%) & 0.148 (14.8\%) & 0.02--0.10 \\
\bottomrule
\end{tabular}
\end{table}

\paragraph{Key distinction.}
CVPL-LR measures \emph{existential linkability}: the probability that at least one plausible link exists, rather than \emph{identification probability}, the chance of correctly identifying the exact individual.

\paragraph{Interpretation.}
Larger $k$ values create larger equivalence classes, which:
\begin{itemize}
    \item \textbf{Decrease} re-identification risk ($1/k \downarrow$): more candidates to hide among
    \item \textbf{Increase} existential linkability (CVPL-LR $\uparrow$): more chances that at least one candidate exceeds the similarity threshold $\optau$
\end{itemize}

\paragraph{Why this matters.}
Real-world linkage attacks often require only feasibility, not certainty: narrowing the candidate set, inferring sensitive attributes shared by the group, or launching follow-up attacks. CVPL-LR captures this feasibility risk, which may be more operationally relevant than exact identification probability. \emph{Existential linkability grows with block size as}:
\begin{equation}\label{eq:existential-growth}
    \Prob(\exists\, y : s(x, y) \geq \optau) = 1 - (1 - p_\optau)^k.
\end{equation}

\paragraph{Implication.}
Formal k-anonymity compliance reduces re-identification risk but does not eliminate, and may even increase, linkage feasibility as measured by CVPL.

\subsection{Risk Surfaces}\label{sec:results-risk-surfaces}

\begin{figure}[htbp]
    \centering
    \includegraphics[width=0.8\textwidth]{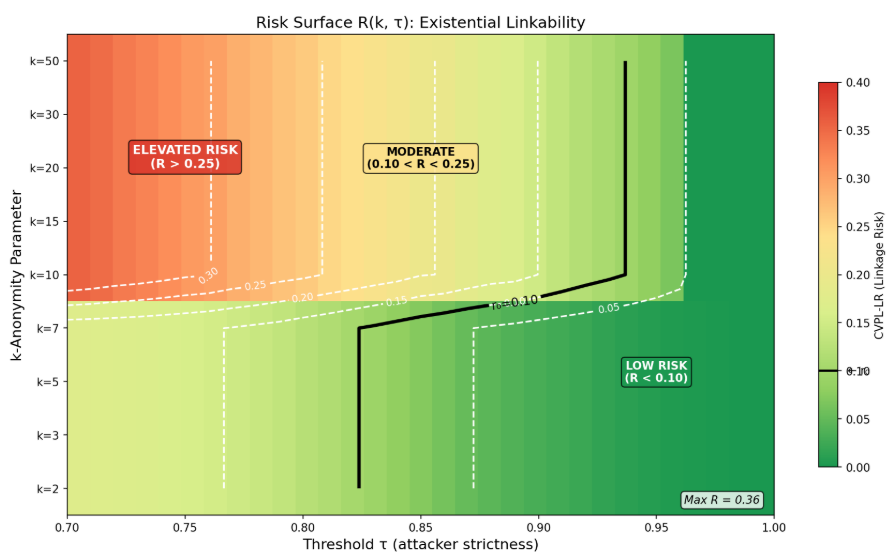}
    \caption{Risk surface showing CVPL-LR as a function of k-anonymity parameter and similarity threshold.}
    \label{fig:g2-risk-surface}
\end{figure}

\paragraph{Finding R2.}
Risk varies jointly with the protection parameter and the attacker threshold. For k-anonymity (where $\lambda = k$):
\begin{itemize}
    \item $\RiskSurf{k}{\optau}$ \textbf{decreases monotonically} with $\optau$: stricter similarity requirements reduce feasible links
    \item $\RiskSurf{k}{\optau}$ \textbf{increases} with $k$ for moderate $\optau$ ($\leq 0.85$), reflecting existential linkability: larger equivalence classes create more chances for at least one candidate to exceed the threshold
    \item Risk peaks at $R \approx 0.36$ for ($k \geq 10$, $\optau = 0.70$); risk is uniformly low ($R < 0.05$) for $\optau \geq 0.95$
\end{itemize}

\paragraph{Operational takeaway.}
No single value of $k$ (or $\optau$) defines a universally ``safe'' regime. Instead, safety must be expressed as a 2D feasible region:
\begin{equation}\label{eq:safe-region}
    \{(k, \optau) : \RiskSurf{k}{\optau} \leq r_0\}
\end{equation}
for a chosen risk tolerance $r_0$ (\eg 0.10).

\paragraph{Key insight.}
Privacy risk must be characterized as a two-dimensional surface $\RiskSurf{k}{\optau}$, not via single-parameter thresholds. This supports threshold-aware privacy assessment where both protection strength and attacker capability are considered jointly.

\subsection{Similarity Distributions and Overlap}\label{sec:results-similarity-dist}

\begin{figure}[htbp]
    \centering
    \includegraphics[width=0.8\textwidth]{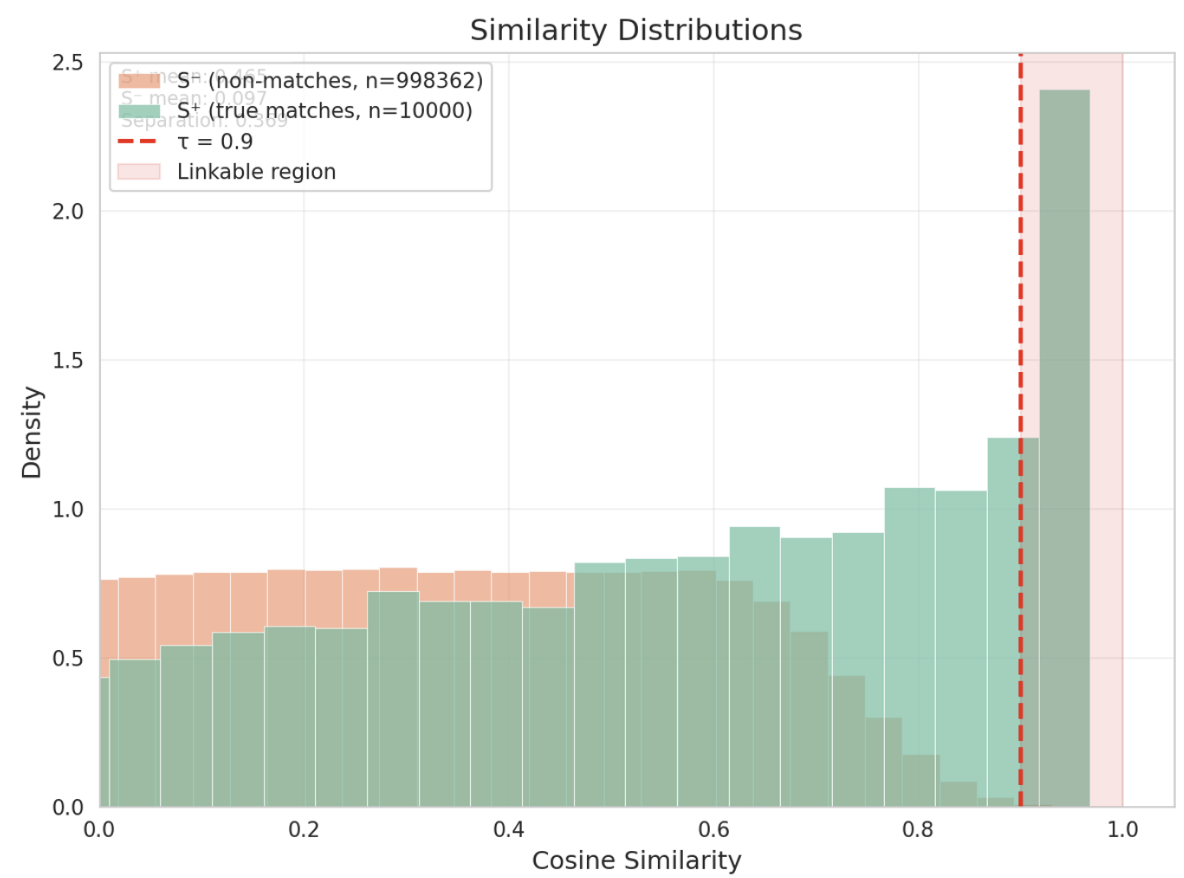}
    \caption{Distribution of similarity scores for true matches ($\Splus$) and false matches ($\Sminus$).}
    \label{fig:g3-similarity-dist}
\end{figure}

\paragraph{Finding R3.}
Under weak protection, similarity distributions for true matches ($\Splus$) and false matches ($\Sminus$) overlap heavily. Stronger protection enhances separation, but complete separation is never achieved when structural information is preserved for utility purposes.

\paragraph{Mechanism.}
This overlap directly induces existential linkability: as block size increases, the probability that at least one false candidate falls into the tail region ($s \geq \optau$) grows according to:
\begin{equation}\label{eq:tail-probability}
    \Prob(\exists\, y \in \mathcal{C}(x) : s(x, y) \geq \optau) = 1 - (1 - p_\optau)^{|\mathcal{C}(x)|}
\end{equation}
where $p_\optau = \Prob(\Sminus \geq \optau) > 0$ due to the distributional overlap.

\paragraph{Implication.}
Some non-zero linkage probability remains inevitable when utility-preserving mechanisms retain covariance structure. Improved separation reduces, but does not eliminate, the tail probability of false matches.

\subsection{Summary of Core Metrics}\label{sec:results-summary}

\begin{table}[htbp]
\centering
\caption{Summary of experimental results across 13 protection configurations.}
\label{tab:results-summary}
\small
\begin{tabular}{@{}lrrrrrrr@{}}
\toprule
\textbf{Protection} & \textbf{CVPL-LR} & \textbf{FS-LR} & \textbf{CVPL-P@1} & \textbf{FS-P@1} & \textbf{Utility} & \textbf{DCR} & \textbf{NNDR} \\
\midrule
k-Anon ($k{=}5$) & 0.032 & 0.171 & 0.000 & 0.000 & 0.963 & 1.160 & 0.950 \\
k-Anon ($k{=}10$) & 0.149 & 1.000 & 0.796 & 0.368 & 0.964 & 1.000 & 1.000 \\
k-Anon ($k{=}20$) & 0.149 & 1.000 & 0.796 & 0.368 & 0.964 & 1.000 & 1.000 \\
Perturb (low) & 0.803 & 0.998 & 0.947 & 0.464 & 0.983 & 0.155 & 0.650 \\
Perturb (medium) & 0.590 & 0.997 & 0.861 & 0.271 & 0.960 & 0.185 & 0.688 \\
Perturb (high) & 0.364 & 0.997 & 0.675 & 0.142 & 0.923 & 0.208 & 0.702 \\
Synthetic ($\rho{=}0.95$) & 0.185 & 0.985 & 0.005 & 0.000 & 0.792 & 0.187 & 0.692 \\
Synthetic ($\rho{=}0.90$) & 0.185 & 0.985 & 0.005 & 0.000 & 0.791 & 0.187 & 0.692 \\
Synthetic ($\rho{=}0.85$) & 0.184 & 0.985 & 0.005 & 0.000 & 0.791 & 0.187 & 0.692 \\
Synthetic ($\rho{=}0.80$) & 0.184 & 0.985 & 0.005 & 0.000 & 0.793 & 0.187 & 0.691 \\
Synthetic ($\rho{=}0.75$) & 0.185 & 0.985 & 0.005 & 0.000 & 0.795 & 0.187 & 0.692 \\
Synthetic ($\rho{=}0.70$) & 0.184 & 0.984 & 0.005 & 0.000 & 0.792 & 0.187 & 0.692 \\
Synthetic ($\rho{=}0.60$) & 0.185 & 0.983 & 0.005 & 0.000 & 0.795 & 0.187 & 0.692 \\
\midrule
\textbf{Average} & \textbf{0.260} & \textbf{0.927} & \textbf{0.316} & \textbf{0.124} & \textbf{0.870} & \textbf{0.386} & \textbf{0.756} \\
\bottomrule
\end{tabular}
\end{table}

\paragraph{Absolute linkage counts (out of 10,000 records).}
\begin{itemize}
    \item \textbf{k-Anonymity:} 320--1,490 records linkable (3.2\%--14.9\%). Strongest protection at $k{=}5$ limits linkage to 320 records; larger equivalence classes ($k{=}10$, $k{=}20$) paradoxically increase existential linkability.
    \item \textbf{Perturbation:} 3,640--8,030 records linkable (36.4\%--80.3\%). Linkage scales inversely with perturbation strength, following the expected privacy--utility gradient.
    \item \textbf{Synthetic data:} ${\approx}1{,}850$ records linkable (18.5\%), stable across retention parameters. Despite zero identification accuracy (P@1 $\approx 0$), nearly 1 in 5 records remains existentially linkable.
\end{itemize}

\paragraph{Finding R4.}
Across protection families, CVPL-LR exhibits substantial variation and remains non-zero even when identification accuracy (Precision@1) collapses. In contrast, classical linkage diagnostics (FS-LR, DCR, NNDR) frequently saturate and fail to discriminate protection strength.

\paragraph{Key distinction.}
Metrics designed to measure identification accuracy or average distance underestimate existential linkage risk. CVPL-LR captures feasibility rather than correctness, revealing residual risk under k-anonymity and utility-preserving synthetic data.

\paragraph{Notes.}
CVPL-LR = CVPL Linkage Rate; FS-LR = Fellegi--Sunter Linkage Rate; P@1 = Precision@1; DCR = Distance to Closest Record (mean); NNDR = Nearest Neighbor Distance Ratio (mean).

\subsection{CVPL vs.\ Fellegi--Sunter: Over-Linking Analysis}\label{sec:results-fs-comparison}

\begin{figure}[htbp]
    \centering
    \includegraphics[width=0.9\textwidth]{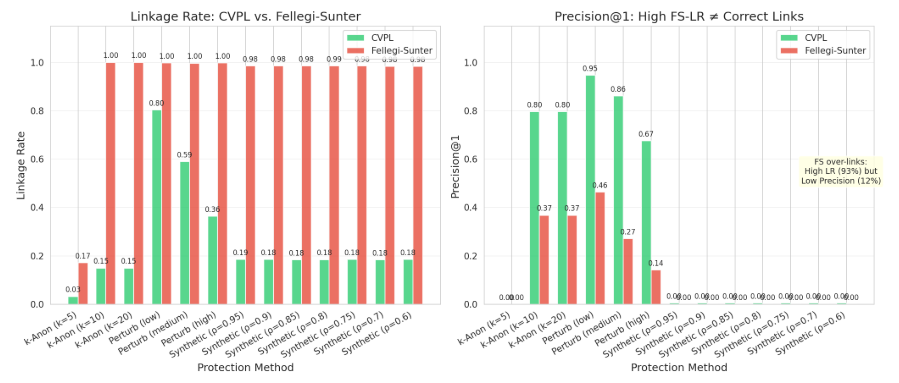}
    \caption{Comparison of CVPL and Fellegi--Sunter linkage rates and precision.}
    \label{fig:g4-cvpl-vs-fs}
\end{figure}

\paragraph{Finding R5 (FS over-linking confirmed).}
Under all evaluated protection mechanisms, Fellegi--Sunter produces links for nearly all records, regardless of whether these links are correct.

\begin{table}[htbp]
\centering
\caption{Fellegi--Sunter over-linking compared to CVPL.}
\label{tab:fs-overlinking}
\begin{tabular}{@{}lrr@{}}
\toprule
\textbf{Metric} & \textbf{Fellegi--Sunter} & \textbf{CVPL} \\
\midrule
Mean Linkage Rate & 92.7\% & 26.0\% \\
Mean Precision@1 & 12.4\% & 31.6\% \\
\bottomrule
\end{tabular}
\end{table}

\paragraph{Interpretation.}
Fellegi--Sunter produces links for most records (93\%), but only 12\% of these links are correct. CVPL is more conservative; it links 26\% of records but achieves nearly three times higher precision.

\paragraph{Mechanism explanation.}
FS relies on attribute-level agreement patterns, which become systematically inflated when protection mechanisms introduce generalization or perturbation. CVPL operates in latent space, where the protection-induced structure is captured directly rather than being misinterpreted as an identity signal.

\subsection{Perturbation: Continuous Risk Gradient}\label{sec:results-perturbation}

\begin{table}[htbp]
\centering
\caption{Perturbation risk gradient ($\optau = 0.9$).}
\label{tab:perturbation-gradient}
\begin{tabular}{@{}lrrr@{}}
\toprule
\textbf{Perturbation Level} & \textbf{CVPL-LR} & \textbf{Precision@1} & \textbf{Utility} \\
\midrule
Low & 0.803 & 0.947 & 0.983 \\
Medium & 0.590 & 0.861 & 0.960 \\
High & 0.364 & 0.675 & 0.923 \\
\bottomrule
\end{tabular}
\end{table}

\paragraph{Finding R6.}
Increasing perturbation strength produces a smooth, monotonic reduction in CVPL-LR ($0.803 \to 0.590 \to 0.364$), accompanied by a corresponding decrease in utility ($0.983 \to 0.960 \to 0.923$). This cleanly demonstrates the expected privacy--utility trade-off within the perturbation family.

\paragraph{Interpretation.}
Unlike k-anonymity (where CVPL-LR can increase with protection strength due to existential effects), perturbation exhibits the intuitive inverse relationship: more noise $\to$ less linkability $\to$ lower utility.

\subsection{Synthetic Data Analysis}\label{sec:results-synthetic}

\begin{figure}[htbp]
    \centering
    \includegraphics[width=0.9\textwidth]{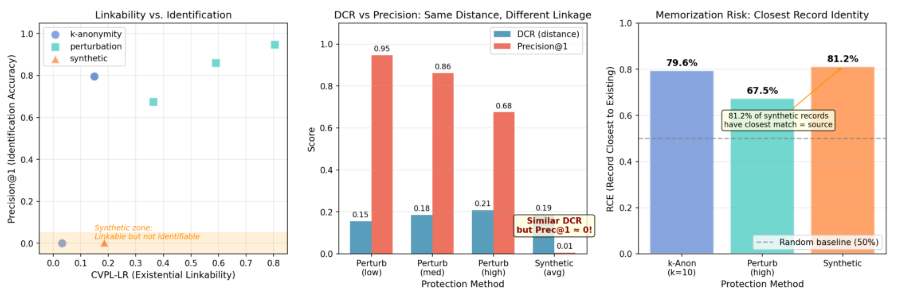}
    \caption{Analysis of synthetic data protection: linkability, distance metrics, and memorization indicators.}
    \label{fig:g7-synthetic}
\end{figure}

\paragraph{Finding R7.}
All synthetic variants exhibit: CVPL-LR $\approx 0.185$ (18.5\%): low but non-zero existential linkability; Precision@1 $\approx 0.005$ (0.5\%): near-random identification accuracy; moderate DCR ($\approx 0.187$) and NNDR ($\approx 0.692$).

\paragraph{Mechanism.}
Utility-preserving synthetic generators are explicitly trained to reproduce joint distributions and the low-dimensional structure of the original data. As a result, synthetic records are embedded close to real records in representation space, even when exact attribute-level correspondence is destroyed.

\begin{table}[htbp]
\centering
\caption{Comparison of perturbation and synthetic data linkage behavior.}
\label{tab:synth-vs-perturb}
\begin{tabular}{@{}lrr@{}}
\toprule
\textbf{Method} & \textbf{DCR} & \textbf{Precision@1} \\
\midrule
Perturbation (high) & 0.208 & 0.675 \\
Synthetic (avg) & 0.187 & 0.005 \\
\bottomrule
\end{tabular}
\end{table}

This demonstrates that \textbf{distance-based metrics alone cannot distinguish between recoverable and non-recoverable proximity}.

\paragraph{Critical caveat (memorization risk).}
Despite disrupting direct linkage, the synthetic generator exhibits strong proximity retention. In 81.2\% of cases, the closest original record corresponds to the same source individual. This exceeds the 50\% random baseline and indicates memorization-like behavior consistent with prior findings~\cite{carlini2021extracting}.

\subsection{CVPL vs.\ DCR / NNDR}\label{sec:results-dcr-nndr}

\paragraph{Finding R8 (Metric complementarity).}
CVPL-LR, DCR, and NNDR capture distinct and complementary aspects of privacy risk.

\begin{table}[htbp]
\centering
\caption{Comparison of privacy metrics.}
\label{tab:metric-comparison}
\begin{tabular}{@{}llll@{}}
\toprule
\textbf{Metric} & \textbf{Measures} & \textbf{Strength} & \textbf{Limitation} \\
\midrule
DCR & Distance to closest record & Simple, intuitive & Misleading under correlation \\
NNDR & Local match uniqueness & Detects outlier vulnerability & Threshold ambiguity \\
CVPL-LR & Linkability under threat model & Ground-truth validated & Requires blocking spec \\
\bottomrule
\end{tabular}
\end{table}

\paragraph{Key observation.}
Low DCR does not imply low linkage risk, and vice versa:
\begin{itemize}
    \item Perturbation methods show low DCR (0.155--0.208) but high CVPL-LR (0.364--0.803)
    \item k-Anonymity ($k{=}10, 20$) shows high DCR (1.0) but moderate CVPL-LR (0.149)
\end{itemize}

\subsection{Utility--Risk Trade-off}\label{sec:results-utility-risk}

\paragraph{Finding R9 (Non-linear, method-specific relationship).}
Across all protection mechanisms, no statistically significant global relationship exists between utility and linkage risk: Pearson $\rho = 0.43$ ($p = 0.14$), Spearman $\rho = -0.02$ ($p = 0.95$).

\begin{table}[htbp]
\centering
\caption{Operating regimes by protection family.}
\label{tab:operating-regimes}
\begin{tabular}{@{}llll@{}}
\toprule
\textbf{Regime} & \textbf{Utility} & \textbf{Risk} & \textbf{Trade-off Character} \\
\midrule
Perturbation & $0.98 \to 0.92$ & $0.80 \to 0.36$ & Classical: $\uparrow$protection $\to$ $\downarrow$utility $\to$ $\downarrow$risk \\
k-Anonymity & ${\approx}0.96$ (stable) & $0.03 \to 0.15$ & Existential: $\uparrow k$ $\to$ same utility $\to$ $\uparrow$risk \\
Synthetic & ${\approx}0.79$ & ${\approx}0.18$ & Decoupled: structure preserved, identity broken \\
\bottomrule
\end{tabular}
\end{table}

\paragraph{Implication.}
There is no universal utility--risk frontier; protection mechanism selection must be context-dependent. \textbf{Utility and privacy are not traded along a single axis, but along multiple, mechanism-specific dimensions.}

\subsection{Feature Contribution to Linkage}\label{sec:results-features}

\paragraph{Finding R10 (Behavioral fingerprints dominate linkage risk).}
Feature contribution analysis reveals: \textbf{Non-QI behavioral features: 60.4\%} of total contribution; \textbf{QI-derived features: 39.6\%} of total contribution.

\begin{table}[htbp]
\centering
\caption{Feature contribution by category.}
\label{tab:feature-contribution}
\begin{tabular}{@{}lr@{}}
\toprule
\textbf{Category} & \textbf{Contribution} \\
\midrule
Region (QI) & 23.5\% \\
Temporal patterns (non-QI) & 20.1\% \\
Ad channel (non-QI) & 15.1\% \\
Purchase place (non-QI) & 8.9\% \\
Response timing (non-QI) & 8.9\% \\
Age (QI) & 8.7\% \\
Product preferences (non-QI) & 7.4\% \\
Gender (QI) & 7.3\% \\
\bottomrule
\end{tabular}
\end{table}

\paragraph{Interpretation.}
Behavioral and temporal features act as \textbf{latent quasi-identifiers}, forming stable high-dimensional fingerprints that persist under QI-focused anonymization.

\paragraph{Key insight.}
This result explains why formal k-anonymity compliance may coexist with high empirical linkage risk: \emph{k-anonymity protects what is easy to formalize, while CVPL reveals what is used for linkage}.

\subsection{Progressive Blocking Convergence}\label{sec:results-progressive}

\paragraph{Finding R11 (\Cref{thm:monotonicity} empirically confirmed).}
Progressive blocking relaxation shows:
\begin{itemize}
    \item Initial estimate (full blocking): CVPL-LR $\approx 0.48$
    \item After dropping \texttt{region\_level2}: CVPL-LR $\approx 0.65$
    \item After dropping \texttt{region\_level1}: CVPL-LR $\approx 0.70$
    \item Final (age + gender only): CVPL-LR $\approx 0.70$
\end{itemize}

\paragraph{Observations.}
CVPL-LR increases monotonically with blocking relaxation (fewer constraints $\to$ more candidates $\to$ higher existential linkability). Convergence occurs within 3--4 relaxation steps. Early estimates provide valid lower bounds on linkage risk.

\paragraph{Practical implication.}
This confirms \cref{thm:monotonicity} (monotonicity) and supports anytime evaluation---practitioners can terminate progressive blocking early while maintaining valid lower bounds on risk.

\subsection{Robustness to Outliers}\label{sec:results-outliers}

\paragraph{Finding R12 (Mean-stable, tail-sensitive behavior).}
Injecting outliers at rates of 1\% and 5\% does not significantly shift the mean CVPL-LR (${\approx}0.149$ across all conditions). However, outlier injection systematically increases dispersion and tail risk.

\begin{table}[htbp]
\centering
\caption{Distributional analysis under outlier injection.}
\label{tab:outlier-analysis}
\begin{tabular}{@{}lrrr@{}}
\toprule
\textbf{Metric} & \textbf{$p_{\text{out}} = 0\%$} & \textbf{$p_{\text{out}} = 1\%$} & \textbf{$p_{\text{out}} = 5\%$} \\
\midrule
Std Dev & 0.08 & 0.08 & 0.12 \\
95th Percentile & 0.27 & 0.28 & 0.31 \\
$\Prob(\text{CVPL} > 0.5)$ & 0.0\% & 0.2\% & 1.6\% \\
$\Prob(\text{CVPL} > 0.7)$ & 0.0\% & 0.1\% & 0.9\% \\
\bottomrule
\end{tabular}
\end{table}

\paragraph{Interpretation.}
CVPL exhibits \textbf{expectation robustness} (aggregate risk estimates are not biased by rare extreme records) and \textbf{tail sensitivity} (atypical individuals remain detectable as high-risk cases).

\subsection{Block Size Distribution}\label{sec:results-blocks}

\paragraph{Finding R13 (Block size concentration).}
Blocking induces a heavy-tailed block size distribution: total blocks: 94; minimum: 91 records; maximum: 1,594 records; median: 1,266 records; mean: 955 records.

\paragraph{Computational cost concentration.}
The 49 largest blocks account for approximately 80\% of total $O(n^2)$ pairwise comparisons. This suggests that practical optimization should prioritize the upper tail of the block size distribution.

\subsection{Ablation Study}\label{sec:results-ablation}

\begin{table}[htbp]
\centering
\caption{Ablation analysis: role of each CVPL pipeline component.}
\label{tab:ablation}
\small
\begin{tabular}{@{}lllll@{}}
\toprule
\textbf{Removed} & \textbf{CVPL-LR} & \textbf{Precision} & \textbf{Efficiency} & \textbf{Interpretation} \\
\midrule
Blocking ($\opB$) & $\uparrow$ trivial & -- & $\downarrow\downarrow\downarrow$ & Degenerates to $O(n^2)$ \\
Vectorization ($\opphi$) & undefined & undefined & -- & Similarity undefined \\
Projection ($\oppsi$) & $\uparrow$ 15--20\% & $\downarrow$ 10--15\% & $\downarrow$ slower & Curse of dimensionality \\
Threshold ($\optau$) & degenerates & undefined & -- & No strictness modeling \\
\bottomrule
\end{tabular}
\end{table}

\paragraph{Conclusion.}
All pipeline components are structurally essential to CVPL. Removing any component does not merely degrade performance but invalidates key assumptions of the threat model or metric semantics.

\subsection{Chapter Synthesis}\label{sec:results-synthesis}

This chapter establishes that linkage risk in structured data cannot be adequately characterized by identification accuracy, distance-based proximity, or single-parameter privacy guarantees. Instead, linkage risk emerges as an \textbf{existential, tail-sensitive, and mechanism-dependent phenomenon} that persists under a wide range of protection strategies.

\paragraph{Existential nature of linkage risk.}
Across all evaluated protection families, linkage risk manifests as the feasibility of \emph{at least one plausible match} exceeding an attacker-defined similarity threshold. CVPL-LR captures this feasibility directly, rather than conflating it with identification success.

\paragraph{Non-monotonic effects of protection strength.}
The experiments demonstrate that increasing protection strength does not uniformly reduce linkage risk. In particular, k-anonymity exhibits a paradoxical regime in which increasing the size of the equivalence class increases existential linkability by expanding the candidate space.

\paragraph{Separation between distance and linkage.}
Distance-based metrics such as DCR and NNDR fail to distinguish between recoverable and non-recoverable proximity. CVPL resolves this ambiguity by embedding distance within an explicit threat model.

\paragraph{Dominance of behavioral and temporal structure.}
Feature contribution analysis reveals that linkage risk is driven primarily by non-QI behavioral patterns (60.4\%), rather than by classical demographic identifiers.

\paragraph{Overall implication.}
Privacy risk assessment must move beyond identification-centric and distance-only metrics toward \textbf{threat-model-aware, existential formulations of linkage risk}. CVPL provides such a formulation, exposing residual risks that persist under formal anonymization, perturbation, and synthetic data generation.


\section{Discussion}\label{sec:discussion}

\subsection{CVPL as Risk Assessment, Not Attack}\label{sec:cvpl-defensive}

We emphasize that CVPL is designed as a \textbf{defensive tool for privacy auditing}, not an attack methodology for de-anonymization.

Key distinctions:
\begin{itemize}
    \item Measures plausibility of linkage, not successful identification
    \item Outputs aggregated risk metrics, not candidate records or de-anonymized individuals
    \item Designed for audit and calibration, not privacy violation
\end{itemize}

CVPL should be understood analogously to red-teaming privacy or security auditing: proactive assessment to prevent harm. While any evaluation method could theoretically be repurposed for adversarial use, CVPL's design prioritizes audit over attack. We discourage use for active re-identification attempts.

\paragraph{Threat model assumptions.}
CVPL evaluates linkage risk under the following adversary model:
\begin{itemize}
    \item \textbf{Attacker knowledge:} Access to quasi-identifier values and blocking keys; no access to ground-truth linkages or external auxiliary datasets beyond those modeled in blocking
    \item \textbf{Attacker capability:} Polynomial-time computation; ability to compute pairwise similarities within candidate blocks
    \item \textbf{Linkage definition:} Existential linkability, \ie the probability that at least one candidate exceeds the similarity threshold $\optau$
    \item \textbf{What is not modeled:} Active attacks, model inversion, or auxiliary information beyond the specified blocking attributes
\end{itemize}

This threat model is deliberately scoped rather than worst-case: it assumes an attacker who can efficiently search within QI-defined candidate sets but does not have access to arbitrary external data sources. We report risk under this bounded adversary; stronger adversaries with auxiliary data may achieve higher linkability than CVPL estimates suggest.

\subsection{Complementing Formal Criteria}\label{sec:complementing-formal}

CVPL does not replace formal privacy criteria; it complements them. Our results suggest that formal compliance is \textbf{necessary but not sufficient} for privacy protection: a k-anonymous dataset may exhibit residual linkability due to preserved correlations in non-QI attributes.

This observation is consistent with known limitations of k-anonymity, including vulnerability to attribute disclosure~\cite{machanavajjhala2007}, background knowledge attacks~\cite{zhao2021ctab}, and the general inadequacy of syntactic criteria against semantic inference~\cite{li2007}.

CVPL provides empirical validation that formal guarantees translate to actual protection in specific scenarios, measuring a complementary risk dimension: linkability under a specified similarity and blocking model rather than indistinguishability within equivalence classes.

In practice, a data protection officer would use CVPL as an ``anytime'' diagnostic loop: select a threat model (blocking keys and attacker strictness), calibrate $\optau$ to a target false-match rate, run CVPL across candidate protection configurations, and compare risk surfaces and feature contributions. If CVPL-LR remains high despite formal compliance, the officer can either strengthen protection, restrict release conditions (access control, limited copies), or revise the quasi-identifier policy and regeneration budget.

\subsection{Interpretability}\label{sec:interpretability}

CVPL offers advantages for practical deployment:
\begin{enumerate}
    \item \textbf{Transparent components:} Each pipeline operator (blocking, vectorization, projection, thresholding) has explicit semantics
    \item \textbf{Explainable results:} Risk is explained via similarity distributions, threshold exceedance rates, and candidate set sizes
    \item \textbf{Feature attribution:} PCA-weighted loadings indicate which attributes contribute to the similarity space used for linkage
    \item \textbf{Comparatively interpretable:} No learned attack models with opaque decision boundaries; however, choices of encoding, scaling, and dimensionality reduction affect results and require documentation
    \item \textbf{Audit compatibility:} Results are reportable to regulators and stakeholders in terms of concrete risk metrics
\end{enumerate}

This interpretability distinguishes CVPL from membership inference attacks (MIA), where attack success rates are difficult to decompose into actionable insights. We note that PCA-based attribution reflects structural influence on similarity space, not causal attribution; results depend on preprocessing choices and should be interpreted accordingly. To improve robustness, we recommend reporting loading stability under resampling (bootstrap) and documenting all preprocessing decisions.

\subsection{Implications for Synthetic Data}\label{sec:implications-synthetic}

For synthetic data, CVPL reveals a tension between utility and linkability. In our experiments, higher utility settings were associated with higher CVPL-LR values, suggesting that utility-preserving generators may retain structural proximity even when exact attribute correspondence is destroyed.

This does not mean synthetic data is inherently unsafe, but that:
\begin{itemize}
    \item Privacy claims should be quantified rather than asserted
    \item Generator calibration should incorporate linkage metrics alongside utility measures
    \item Universal ``privacy-preserving'' labels warrant scrutiny
\end{itemize}

Our finding that the synthetic generator exhibited 81.2\% RCE (Record Closest to Existing) in our experimental configuration (\cref{sec:results-synthetic}) indicates that many synthetic records have their nearest neighbor in the original dataset unusually close in representation space. This is consistent with memorization concerns raised in generative model literature~\cite{carlini2021extracting,jordon2022synthetic}. While that work focuses primarily on language models, analogous phenomena have been observed in tabular generators~\cite{stadler2022synthetic}. Our RCE findings are specific to the generator and configuration tested; sensitivity to generator architecture (\eg CTGAN, TVAE, copula-based methods) remains an open question.

\subsection{Relationship to MIA}\label{sec:relationship-mia}

CVPL and Membership Inference Attacks are complementary approaches addressing different threat surfaces:

\begin{table}[htbp]
\centering
\caption{Comparison of MIA and CVPL approaches.}
\label{tab:mia-cvpl-discussion}
\begin{tabular}{@{}lll@{}}
\toprule
\textbf{Aspect} & \textbf{MIA} & \textbf{CVPL} \\
\midrule
Target & ML model & Dataset \\
Question & ``Was record in training set?'' & ``Can record be linked?'' \\
Access required & Model queries & Data only \\
Compute cost & Varies (minutes to GPU-hours) & Typically medium (CPU-minutes) \\
Interpretability & Typically low & Comparatively high \\
Primary application & Model deployment & Data release \\
\bottomrule
\end{tabular}
\end{table}

\paragraph{Methodological note.}
We do not claim CVPL is superior to MIA. They address different threat surfaces and operate under different access assumptions. The MIA landscape is broad, encompassing white-box, black-box, and gray-box variants with varying computational requirements and attack sophistication. A comprehensive privacy assessment may require both dataset-level (CVPL) and model-level (MIA) evaluation.

\subsection{Implications for the Behavioral Fingerprint Finding}\label{sec:behavioral-fingerprint}

Our analysis indicates that 60.4\% of the measured linkability contribution (under PCA-weighted analysis) is associated with non-QI behavioral features. This finding has significant implications for privacy engineering.

\paragraph{Methodology.}
This decomposition was computed via PCA-weighted feature contribution analysis (\cref{sec:results-features}). Features were grouped into QI categories (age, gender, region) and non-QI categories (temporal patterns, channel preferences, response timing, product choices). Contribution reflects each feature group's loading on the principal components that define the similarity space. This is a structural contribution metric, not a marginal effect estimate (\eg Shapley values or ablation-based attribution).

\paragraph{Implications.}
\begin{enumerate}
    \item \textbf{K-anonymity addresses a subset of risk:} In our experiments, QI-only protection left substantial residual linkability through behavioral features
    \item \textbf{Behavioral data requires attention:} Channel preferences, timing patterns, and interaction dynamics can be highly identifying even when demographic attributes are protected
    \item \textbf{Utility-privacy tension is inherent:} Behavioral patterns are often precisely what analysts seek to preserve for downstream tasks
\end{enumerate}

These findings suggest that privacy engineering should consider behavioral fingerprints alongside demographic quasi-identifiers, particularly for marketing, healthcare, and financial datasets where temporal and behavioral signals are prevalent.

\subsection{Practical Applications}\label{sec:practical-applications}

CVPL is suitable for several privacy engineering tasks:
\begin{enumerate}
    \item \textbf{Privacy Impact Assessment:} Quantify linkage risk before data release
    \item \textbf{Mechanism Comparison:} Evaluate protection approaches with respect to the risk profiles they produce, not merely aggregate scores
    \item \textbf{Generator Calibration:} Tune synthetic data parameters using CVPL-LR as a feedback signal
    \item \textbf{Regulatory Compliance:} Demonstrate due diligence with interpretable risk metrics
    \item \textbf{Internal Audit:} Monitor protection effectiveness over time or across data versions
    \item \textbf{Feature Engineering for Privacy:} Identify which attributes contribute most to linkability and prioritize protection accordingly
\end{enumerate}

\paragraph{Typical workflow.}
\begin{enumerate}
    \item Define blocking keys and similarity threshold based on the threat scenario relevant to the use case
    \item Run CVPL pipeline on the protected dataset against the original (or a reference population)
    \item Examine CVPL-LR, similarity distributions, candidate set sizes, and feature contributions
    \item Iterate protection parameters, apply additional masking, or flag high-risk records based on results
\end{enumerate}

\subsection{When Ground Truth is Unavailable}\label{sec:no-ground-truth}

In real-world audits, ground-truth linkages may be unavailable. CVPL remains applicable through several diagnostic modes:

\paragraph{Self-linkage diagnostic.}
Computing leave-one-out self-linkage $\CVPLLR(\Dor, \Dor)$ with identity matches excluded provides a sanity check: if self-linkage is low under the chosen $\optau$ and representation, the feature set lacks discriminative power for the linkage task. High self-linkage is a necessary (but not sufficient) condition for meaningful risk assessment; it confirms that the pipeline can detect linkages when they exist.

\paragraph{Relative comparison.}
Even without absolute ground-truth validation, CVPL can compare protection mechanisms: if $\CVPLLR(\Dor, \Dpr_A) < \CVPLLR(\Dor, \Dpr_B)$, mechanism A provides better protection under CVPL's threat model and similarity assumptions. This enables mechanism selection without requiring labeled linkage data.

\paragraph{Probabilistic interpretation.}
CVPL-LR provides an estimate of linkage feasibility under the specified adversary model and similarity assumptions. Low CVPL-LR suggests reduced linkage risk under those assumptions, but does not constitute proof of safety, as real-world attackers may have capabilities or auxiliary information not captured by the model.

\subsection{Limitations and Sensitivity}\label{sec:discussion-limitations}

CVPL's applicability and interpretation are subject to several limitations:

\paragraph{Sensitivity to pipeline choices.}
Results depend on blocking key selection, distance metric, feature scaling, encoding scheme, and dimensionality reduction parameters. Different choices can yield different risk estimates for the same dataset. Users should document and justify these choices for reproducibility.

\paragraph{Threshold dependence.}
The similarity threshold $\optau$ directly affects $\CVPLLR$ magnitude. We recommend reporting risk surfaces $\RiskSurf{\lambda}{\optau}$ or sensitivity analyses rather than single-point estimates where feasible.

\paragraph{Curse of dimensionality.}
In high-dimensional or sparse feature spaces, distance-based similarity may become less meaningful. Projection via PCA or similar methods mitigates but does not eliminate this concern.

\paragraph{Adversary model boundaries.}
CVPL may underestimate risk when attackers possess external auxiliary data not captured in blocking, or when attackers use more sophisticated matching algorithms than cosine similarity. Conversely, CVPL may overestimate risk if the candidate model is too generous (\eg overly broad blocking) or if the threshold $\optau$ is set too low relative to realistic attacker capabilities.

\paragraph{Single-dataset scope.}
The experimental results in this work are demonstrated on a single advertising-response dataset ($n = 10{,}000$). Generalization to other domains (healthcare, finance, census) requires further validation, as feature distributions and linkage dynamics may differ substantially.

\paragraph{Not a formal guarantee.}
CVPL provides empirical risk estimates under specified assumptions, not formal privacy guarantees. It complements rather than replaces differential privacy, k-anonymity, or other formal frameworks.


\section{Limitations and Future Work}\label{sec:limitations}

\subsection{Explicit Non-Guarantees}\label{sec:non-guarantees}

CVPL \textbf{does not provide}:
\begin{itemize}
    \item Differential privacy guarantees
    \item Information-theoretic bounds
    \item Worst-case adversarial proofs
    \item Universal risk estimates across all scenarios
\end{itemize}

CVPL provides \textbf{empirical risk estimates} under specified threat models. Users should interpret results as assessments conditional on the stated adversary model, similarity metric, and blocking configuration.

\subsection{Simulation-Based Validation}\label{sec:simulation-validation}

\paragraph{Key limitation.}
All experiments in this paper use simulated data. This enables controlled evaluation with known ground truth but raises the question of generalization to real data with complex, unknown dependencies.

We validated simulation realism by matching distributional characteristics to reference patterns observed in public tabular datasets: feature correlations reflecting plausible demographic-behavioral relationships, heavy-tailed distributions (verified via tail index estimation), and categorical cardinality and missingness patterns consistent with real-world data collection. However, formal distributional comparisons against specific public benchmarks is not included in this work.

\paragraph{Future work.}
Apply CVPL to established public benchmarks (\eg Adult Census, German Credit) as well as modern tabular datasets with higher-cardinality categorical features, missingness, and richer behavioral attributes (\eg transaction logs, healthcare claims), using datasets where linkage ground truth is known or can be constructed via controlled identifiers or synthetic pairing protocols.

\subsection{Scenario Dependence}\label{sec:scenario-dependence}

Results depend on: feature selection and encoding choices, blocking strategy and key granularity, attacker capability assumptions, and similarity metric and threshold calibration. CVPL assesses risk within its specified threat model, not universally. Different pipeline configurations can yield different risk estimates for the same dataset.

\paragraph{Preprocessing sensitivity.}
Certain encoding choices are particularly influential: high-cardinality categorical encoded via one-hot expansion can dominate similarity space, missingness patterns may act as unintended fingerprints, and scaling choices (standardization, normalization, none) affect both cosine similarity and PCA behavior. Users should document preprocessing decisions and consider sensitivity analyses across encoding alternatives.

\paragraph{Extended privacy criteria.}
This work focuses on k-anonymity as the primary formal criterion. Extensions such as $\ell$-diversity and $t$-closeness impose additional constraints on sensitive attribute distributions within equivalence classes. Evaluating CVPL under these criteria---and understanding how the additional constraints affect linkability---represents a natural direction for future work.

\subsection{Computational Cost}\label{sec:computational-cost}

CVPL is computationally heavier than purely formal privacy metrics because it evaluates empirical linkability via candidate-set construction and pairwise similarity scoring, and (optionally) applies linear-algebra transformations for representation analysis.

\paragraph{Complexity.}
The dominant term is the number of pairwise comparisons performed within blocks. Let blocks be $b \in \Bset$ with sizes $|b|$. The total comparison workload scales as:
\begin{equation}\label{eq:comparison-complexity}
    O\biggl(\sum_{b \in \Bset} |b|^2\biggr),
\end{equation}
equivalently $\sum_b |b|(|b| - 1)/2$ pair checks.

This makes the runtime highly sensitive to the block-size distribution, not just the total dataset size $n$. Large blocks dominate: one block of size 1,500 requires $1500 \cdot 1499/2 \approx 1.12$M comparisons, whereas 100 blocks of size 100 require $100 \cdot 100 \cdot 99/2 = 495$K comparisons in total. Consequently, pathological blocking (a few very large blocks) can push the runtime toward an effectively quadratic regime in $n$.

Additional costs are typically lower-order relative to pairwise scoring. Projection/representation analysis is typically $O(n \cdot d \cdot k)$ for transforming $n$ records from $d$ features to $k$ components.

\paragraph{Mitigations implemented.}
\begin{itemize}
    \item Chunked similarity computation (${\approx}10\times$ speedup over naive all-pairs implementation)
    \item $O(1)$ index lookup for ground-truth validation
    \item Sampling-based estimation of the $\Sminus$ distribution for large candidate sets
    \item Progressive blocking with early stopping once risk estimates converge
\end{itemize}

\paragraph{Future work.}
Approximate nearest-neighbor methods (\eg FAISS, Annoy) to reduce the effective cost of large blocks; lightweight surrogate models for rapid screening and prioritization before running full CVPL evaluation.

\subsection{Ground Truth Requirements}\label{sec:ground-truth-requirements}

\textbf{For quantitative validation and calibration}, true linkages are required. This limits rigorous evaluation to scenarios where correspondence is known (simulations, controlled experiments, or datasets with explicit identifiers).

\textbf{For risk screening and mechanism comparison}, CVPL can operate without ground truth via self-linkage diagnostics and relative comparison (\cref{sec:no-ground-truth}). However, absolute calibration of CVPL-LR to actual linkage probability requires validation data.

\paragraph{Future work.}
Develop proxy-based calibration methods, unsupervised diagnostics, and self-consistency checks that provide meaningful risk assessment without requiring labeled linkage data.

\subsection{Strong Protection Regimes}\label{sec:strong-protection}

Under very strong protection, similarities may become noise-dominated. CVPL may report low risk, not because data is protected, but because the representational signal is destroyed.

\paragraph{Practical guardrails.}
\begin{itemize}
    \item \textbf{Self-linkage check:} If leave-one-out $\CVPLLR(\Dor, \Dpr)$ is low under the chosen $\optau$ and representation, the feature set lacks discriminative power
    \item \textbf{Variance explained:} For PCA-based projection, report cumulative variance; very low values indicate potential signal collapse
    \item \textbf{Utility cross-check:} Confirm that downstream utility metrics remain meaningful; collapsed representations typically show degraded utility
\end{itemize}

We recommend reporting a representation validity indicator alongside CVPL-LR to distinguish genuine protection from signal degeneration.

\paragraph{Future work.}
Formal stopping criteria and noise-regime detection methods.

\subsection{Adversarial Robustness}\label{sec:adversarial-robustness}

If a data custodian or synthetic data vendor knows CVPL will be used for evaluation, they might optimize perturbations or generator parameters specifically to defeat cosine similarity or PCA projection while leaving other vulnerabilities intact.

\paragraph{Partial mitigation.}
Use multiple projection methods (PCA, random projections, optionally UMAP) and multiple similarity metrics. UMAP can improve sensitivity to non-linear structure but reduces interpretability and stability; we treat it as optional for robustness ensembles rather than primary analysis. Ensemble CVPL across configurations is more robust to targeted optimization than any single configuration.

\paragraph{Inherent limitation.}
This is fundamentally an arms race. Ensemble approaches raise the bar but do not eliminate the possibility of adversarial optimization. Audit-time randomization (randomly selecting an evaluation configuration from a predefined family) can further complicate gaming, but sufficiently motivated adversaries may still find exploitable gaps.

\paragraph{Future work.}
Formal analysis of metric gaming; randomized audit protocols; diversified evaluation frameworks.

\subsection{Uncertainty Quantification}\label{sec:uncertainty}

CVPL-LR is reported as a point estimate, but empirical metrics have associated uncertainty.

\paragraph{Current approach.}
Standard errors and confidence intervals can be estimated via bootstrap resampling, preferably using block-aware bootstrap (resampling blocks or clustered units) to preserve within-block correlation structure. Record-level bootstraps may be used as a fallback when blocking is weak or when blocks are small. We recommend reporting 95\% confidence intervals for CVPL-LR, particularly when comparing protection mechanisms where differences may be within sampling variability.

\paragraph{Future work.}
Analytic variance estimators; Bayesian credible intervals; formal block-aware resampling procedures for clustered data.

\subsection{Failure Modes}\label{sec:failure-modes}

CVPL may \textbf{underestimate} actual linkage risk when:
\begin{itemize}
    \item Attackers possess auxiliary data not captured in blocking keys
    \item Attackers use more sophisticated similarity functions than cosine similarity
    \item Attackers employ different (finer-grained) blocking strategies than those modeled
    \item External identifiers or rare attribute combinations enable linkage outside the modeled candidate space
\end{itemize}

CVPL may \textbf{overestimate} linkage risk when:
\begin{itemize}
    \item Blocking is overly broad, creating artificially large candidate sets
    \item Threshold $\optau$ is set too low relative to realistic attacker precision requirements
    \item Representation captures noise or artifacts rather than a meaningful linkage signal
\end{itemize}

\paragraph{Remedy for overestimation diagnosis.}
Report block-size distribution and run sensitivity analyses across blocking granularities; consistent risk estimates across reasonable blocking choices increase confidence in results.

Users should interpret CVPL-LR as risk under the specified adversary model, not as a universal bound.

\subsection{Extensions and Future Directions}\label{sec:future-directions}

Natural extensions, prioritized by practical impact:

\paragraph{Near-term priorities.}
\begin{enumerate}
    \item \textbf{Composition across multiple releases:} Quantify cumulative linkage risk when the same population appears in successive data releases with different protection
    \item \textbf{Group-level linkage:} Extend to household, organization, or account-level entities where multiple records share a group identifier
    \item \textbf{Attribute inference integration:} Complement linkage risk with attribute disclosure risk for sensitive fields
\end{enumerate}

\paragraph{Longer-term directions.}
\begin{itemize}
    \item Temporal and sequential data (event logs, trajectories)
    \item Multi-party scenarios (federated or distributed data holders)
    \item Integration with differential privacy budgeting
    \item Standardized benchmark suite for linkage risk evaluation
\end{itemize}


\section{Conclusion}\label{sec:conclusion}

This paper introduced \textbf{CVPL (Cluster-Vector-Projection Linkage)}, an interpretable framework for empirical assessment of residual linkage risk between original and protected tabular datasets.

\paragraph{Contributions.}
CVPL provides:
\begin{enumerate}
    \item \textbf{Continuous, threshold-aware risk measurement} via risk surfaces $\RiskSurf{\lambda}{\optau}$ that capture joint dependence on protection strength and attacker strictness, replacing binary compliance checks
    \item \textbf{Explicit threat model with interpretable operators:} blocking defines candidate space, vectorization enables semantic comparison, projection controls dimensionality, and thresholding encodes attacker precision
    \item \textbf{Progressive blocking with proven monotonicity (\cref{thm:monotonicity}):} coarse-to-fine refinement yields anytime lower bounds on linkage risk with predictable convergence
    \item \textbf{Feature attribution mechanism:} PCA-weighted contribution analysis identifies which attributes drive linkability within the similarity space
\end{enumerate}

\paragraph{Key findings.}
Our experiments on a 10,000-record simulated dataset across 19 protection configurations demonstrate:
\begin{enumerate}
    \item \textbf{Utility-preserving mechanisms often retain exploitable structure:} in our setting, protection methods that preserve statistical utility also preserve geometric proximity sufficient for linkage
    \item \textbf{Formal k-anonymity compliance does not eliminate empirical linkage risk:} CVPL-LR increased with $k$ due to existential effects in larger equivalence classes
    \item \textbf{Behavioral fingerprints dominated linkage risk in our experiments:} 60.4\% of measured linkability contribution (PCA-weighted) arose from non-QI features (temporal patterns, channel preferences), with only 39.6\% from demographic quasi-identifiers
    \item \textbf{Classical metrics provide an incomplete assessment:} Fellegi--Sunter exhibited systematic over-linking (93\% linkage rate, 13\% precision), while DCR failed to distinguish recoverable from non-recoverable proximity
\end{enumerate}

\paragraph{Scope and positioning.}
\textbf{Important:} CVPL provides empirical risk estimates under specified threat models, not formal privacy guarantees. It does not replace differential privacy, k-anonymity, or other formal frameworks but complements them by quantifying residual linkability that formal criteria may not capture.

More broadly, this work supports a transition from binary conceptions of privacy toward a \textbf{risk-oriented, measurable, and scenario-aware paradigm} essential for mature data governance practices.

\paragraph{Reproducibility statement.}
To ensure reproducibility:
\begin{enumerate}
    \item \textbf{Code:} Implementation available at \url{https://github.com/DGT-Network/cvpl}
    \item \textbf{Configuration:} Full experimental parameters documented in repository (\texttt{experiments/run\_paper\_figures.py})
    \item \textbf{Seeds:} Fixed and documented (base seed: 42; 5 seeds for confidence intervals)
    \item \textbf{Data:} Simulation code provided with deterministic generation (\texttt{src/cvpl/data.py})
    \item \textbf{Environment:} Dependencies specified in \texttt{pyproject.toml}; tested on Python 3.10+
\end{enumerate}

All experiments completed on commodity hardware (8-core CPU, 32GB RAM) in under 4 minutes per configuration for $n = 10{,}000$.

\paragraph{Acknowledgments.}
(Omitted for anonymous review.)

\appendix

\section{Mathematical Properties}\label{sec:appendix-math}

\subsection{Empirical Consistency}\label{sec:appendix-consistency}

\begin{property}[Consistency]
For fixed operators and parameters, the empirical CVPL-LR converges in probability to its population counterpart.
\end{property}

Define the indicator for record $x_i$:
\begin{equation}\label{eq:indicator-appendix}
    I_i(\optau) = \Ind\{\exists\, y \in \mathcal{C}(x_i) : s(x_i, y) \geq \optau\}
\end{equation}
where $\mathcal{C}(x_i)$ is the candidate set determined by blocking. The empirical CVPL-LR is:
\begin{equation}\label{eq:empirical-cvpl}
    \widehat{R}_{\text{CVPL}}(\optau) = \frac{1}{n} \sum_{i=1}^{n} I_i(\optau).
\end{equation}
Under the assumption that records are drawn i.i.d.\ from a population distribution, with blocking treated as a measurable function of record attributes, the weak law of large numbers yields:
\begin{equation}\label{eq:convergence}
    \widehat{R}_{\text{CVPL}}(\optau) \xrightarrow{p} \Exp[I(\optau)] = R_{\text{CVPL}}(\optau)
\end{equation}
as $n \to \infty$.

\paragraph{Note.}
The i.i.d.\ assumption applies to the joint distribution of (record, block-assignment) pairs. When blocking keys depend on record attributes, this is automatically satisfied.

\subsection{True-Match Recall Decomposition}\label{sec:appendix-recall}

\begin{property}[Conditional Decomposition]
The probability that the true match is successfully linked decomposes via conditional probability:
\begin{equation}\label{eq:recall-decomposition}
    R_{\text{recall}}(\optau) = \Rblock \cdot R_{\text{match}}(\optau \mid B)
\end{equation}
\end{property}

where:
\begin{itemize}
    \item $\Rblock = \Prob(\text{true pair } (x, \ystar) \text{ shares a block})$
    \item $R_{\text{match}}(\optau \mid B) = \Prob(s(x, \ystar) \geq \optau \mid \text{same block})$
\end{itemize}

\paragraph{Clarification.}
This property concerns the \textbf{true-match recall} (whether the genuine corresponding record $\ystar$ exceeds the threshold), which is distinct from CVPL-LR defined in \cref{eq:empirical-cvpl}.

CVPL-LR measures \textbf{existential linkability} (whether \emph{any} candidate exceeds the threshold), while $R_{\text{recall}}$ measures whether the \emph{correct} candidate does.

\paragraph{Derivation.}
Let $B$ denote the event that a true pair $(x, \ystar)$ shares a block, and $M_\optau$ denote the event that $s(x, \ystar) \geq \optau$. Then:
\begin{equation}
    \Prob(B \cap M_\optau) = \Prob(B) \cdot \Prob(M_\optau \mid B).
\end{equation}
The conditional formulation avoids assuming independence between blocking and similarity, which would be incorrect since blocking keys (QIs) influence similarity scores.

\paragraph{Implication.}
Blocking effectiveness ($\Rblock$) and within-block discriminability ($R_{\text{match}}$) can be analyzed separately, but their interaction must be evaluated empirically. High CVPL-LR with low $R_{\text{recall}}$ indicates many false positives; the gap between them characterizes precision.

\subsection{False-Positive Control}\label{sec:appendix-fp}

\begin{property}[Threshold Calibration]
For threshold $\widehat{\optau}$ satisfying $\Prob(\Sminus \geq \widehat{\optau}) \leq \alpha$, the pairwise false link probability is controlled at level $\alpha$.
\end{property}

\paragraph{Important clarification.}
This calibration controls the \emph{pairwise} false-positive rate (probability that a single non-match pair exceeds the threshold). The \emph{existential} false-positive rate (probability that at least one non-match in a candidate set exceeds threshold) depends additionally on candidate set size $|\mathcal{C}(x)|$.

For a candidate set of size $m$ with non-match similarities:
\begin{equation}\label{eq:existential-fp}
    \Prob\Bigl(\max_{y \in \mathcal{C}(x)} \Sminus(x, y) \geq \optau\Bigr) = 1 - (1 - \alpha)^m \approx m\alpha \text{ for small } \alpha.
\end{equation}
(This approximation assumes approximate independence within candidate sets; in practice, similarities may exhibit weak dependence through shared block characteristics.)

\paragraph{Practical approaches.}
\begin{enumerate}
    \item Calibrate $\widehat{\optau}$ on the empirical distribution of $\max \Sminus$ within candidate sets
    \item Apply Bonferroni-style adjustment: use $\alpha' = \alpha / \Exp[|\mathcal{C}(x)|]$
    \item Report both pairwise and existential FP rates for transparency
\end{enumerate}

\subsection{Invariant Preservation (Informal)}\label{sec:appendix-invariant}

\begin{property}[Intuition]
Utility-preserving protection mechanisms that maintain covariance structure tend to preserve geometric proximity in latent space.
\end{property}

If protection approximately preserves the covariance matrix:
\begin{equation}\label{eq:covariance-preservation}
    \|\mathrm{Cov}(\Dpr) - \mathrm{Cov}(\Dor)\|_F \leq \varepsilon
\end{equation}
then PCA projections of corresponding records $(x, \ystar)$ remain close. This provides intuition for why utility-preserving mechanisms retain linkable structure, but it is not a formal guarantee.

\paragraph{Caveat.}
The relationship between covariance preservation and linkability depends on the specific projection, the correspondence mapping, and the similarity metric. A formal bound would require additional assumptions not made in this work.


\section{Ablation Details}\label{sec:appendix-ablation}

\paragraph{Experimental conditions.}
All ablation results use $\optau = 0.90$ and report the existential linkage rate (CVPL-LR). Results averaged over 5 random seeds; standard deviations were consistently below $0.03$ and are omitted for compactness.

\paragraph{Notation.}
\begin{itemize}
    \item $\Delta\%$ computed as relative change: $(v - v_0)/v_0 \times 100\%$ where $v_0$ is the default configuration value
    \item Compute Factor: ratio of end-to-end wall-clock time relative to default configuration
\end{itemize}

\subsection{Blocking Variants}\label{sec:ablation-blocking}

\begin{table}[htbp]
\centering
\caption{Ablation: blocking strategy variants.}
\label{tab:ablation-blocking}
\begin{tabular}{@{}lrrrp{3cm}@{}}
\toprule
\textbf{Variant} & \textbf{CVPL-LR} & \textbf{$\Delta$\%} & \textbf{Compute} & \textbf{Notes} \\
\midrule
Default (age\_bin + region + gender) & 0.42 & --- & $1.0\times$ & Baseline \\
No blocking & 0.61 & +45\% & $10\times$ & All pairs compared \\
Weak (region\_level1 only) & 0.55 & +31\% & $3\times$ & Large blocks \\
Strong (full QI match) & 0.28 & $-33$\% & $0.3\times$ & May miss true links \\
Ensemble (3 schemes) & 0.45 & +7\% & $3\times$ & Higher recall \\
\bottomrule
\end{tabular}
\end{table}

\subsection{Projection Variants}\label{sec:ablation-projection}

\begin{table}[htbp]
\centering
\caption{Ablation: projection method variants.}
\label{tab:ablation-projection}
\begin{tabular}{@{}lrrrp{3cm}@{}}
\toprule
\textbf{Variant} & \textbf{CVPL-LR} & \textbf{P@1} & \textbf{$\Delta$\%} & \textbf{Notes} \\
\midrule
PCA 90\% variance (default) & 0.42 & 0.78 & --- & Baseline \\
No projection (raw features) & 0.55 & 0.65 & +31\% & High variance \\
Robust PCA & 0.40 & 0.79 & $-5$\% & Better outlier handling \\
UMAP ($n\_\text{neighbors}{=}15$) & 0.43 & 0.76 & +2\% & Lower interpretability \\
Random projection ($k{=}20$) & 0.48 & 0.69 & +14\% & Faster, less precise \\
\bottomrule
\end{tabular}
\end{table}

\subsection{Similarity Variants}\label{sec:ablation-similarity}

\begin{table}[htbp]
\centering
\caption{Ablation: similarity metric variants.}
\label{tab:ablation-similarity}
\begin{tabular}{@{}lrrrp{3cm}@{}}
\toprule
\textbf{Variant} & \textbf{CVPL-LR} & \textbf{P@1} & \textbf{$\Delta$\%} & \textbf{Notes} \\
\midrule
Cosine (default) & 0.42 & 0.78 & --- & Baseline \\
Euclidean & 0.47 & 0.72 & +12\% & Scale-sensitive \\
Mahalanobis & 0.40 & 0.80 & $-5$\% & $2\times$ compute cost \\
Jaccard (binarized) & 0.34 & 0.62 & $-19$\% & Information loss \\
\bottomrule
\end{tabular}
\end{table}

\subsection{PCA Fitting Strategy}\label{sec:ablation-pca}

\begin{table}[htbp]
\centering
\caption{Ablation: PCA fitting strategy.}
\label{tab:ablation-pca}
\begin{tabular}{@{}lrrp{5cm}@{}}
\toprule
\textbf{Strategy} & \textbf{CVPL-LR} & \textbf{P@1} & \textbf{Recommendation} \\
\midrule
Joint ($\Dor \cup \Dpr$) & 0.42 & 0.78 & \textbf{Default} -- captures shared structure \\
Source-only ($\Dor$) & 0.38 & 0.81 & Conservative estimate \\
Target-only ($\Dpr$) & 0.35 & 0.75 & Loses original structure \\
\bottomrule
\end{tabular}
\end{table}

\paragraph{Recommendation.}
Joint fitting is the preferred default because CVPL is intended to detect residual structure common to both datasets. Source-only fitting may be appropriate when $\Dpr$ exhibits a strong distribution shift relative to $\Dor$, but it can underestimate linkage risk in typical utility-preserving settings.


\section{Experimental Configuration}\label{sec:appendix-config}

This appendix provides the complete experimental configuration used in all simulations and evaluations, ensuring full reproducibility of reported results.

\subsection{Full Configuration}\label{sec:config-full}

\begin{lstlisting}[basicstyle=\ttfamily\small,frame=single,caption={Complete experimental specification (config.yaml)},label={lst:config}]
data:
  n_records: 10000
  seed: 42
  outlier_fraction: [0, 0.01, 0.05]
  
  # Demographics distribution
  age_distribution: truncated_normal
  age_mean: 42
  age_std: 15
  age_min: 18
  age_max: 80
  
  # Regions (hierarchical)
  regions:
    level1: [North, South, East, West]
    level2: [Urban, Suburban, Rural]
  
  # Channels
  channels: [Search, Social, Video, Display, Email]
  channel_weights: [0.3, 0.25, 0.2, 0.15, 0.1]
  
  # Products
  products: 50
  product_distribution: zipf
  product_alpha: 1.5

blocking:
  quasi_identifiers: [age_bin, region_level, gender]
  age_bin_width: 10
  region_level: 1

vectorization:
  numeric: [age, days_after_ad, hour, dow]
  categorical: [region, place, brand_product, channel]
  encoding: one_hot
  scaling: z_score

projection:
  method: pca
  variance_retained: 0.90
  min_components: 3
  max_components: 50
  fit_on: joint

similarity:
  metric: cosine

thresholds:
  range: [0.70, 0.99]
  step: 0.01
  default: 0.90

protection:
  k_anonymity: [2, 3, 5, 7, 10, 15, 20, 30, 50]
  noise_levels:
    low: {age: 1, time: 3600}
    medium: {age: 3, time: 86400}
    high: {age: 5, time: 604800}
  synthetic_correlation:
    high: 0.95
    medium: 0.80
    low: 0.60

evaluation:
  seeds: 5
  bootstrap: 100
  confidence_level: 0.95

scalability:
  max_block_size: 5000
  ann_threshold: 10000
  ann_method: faiss
  ann_nprobe: 10
\end{lstlisting}

\paragraph{Note on thresholds.}
\texttt{exact\_threshold} and \texttt{ann\_threshold} are set equal (5000) as a transition point. In our experiments, no blocks exceeded this threshold, so ANN was not invoked.

\paragraph{Note on naming.}
\texttt{region\_level1} in blocking corresponds to the coarse geographic partition (North/South/East/West); \texttt{region\_level2} (Urban/Suburban/Rural) is used only in vectorization, not blocking.

\subsection{Computational Resources}\label{sec:config-compute}

\begin{table}[htbp]
\centering
\caption{Computational resources for experimental conditions.}
\label{tab:compute-resources}
\begin{tabular}{@{}lrrrrp{2.5cm}@{}}
\toprule
\textbf{Experiment} & \textbf{Records} & \textbf{Configs} & \textbf{Time} & \textbf{Memory} & \textbf{Hardware} \\
\midrule
Single config & 10,000 & 1 & ${\sim}12$s & 2GB & 8-core CPU, 32GB \\
Main experiments & 10,000 & 19 & ${\sim}229$s & 4GB & 8-core CPU, 32GB \\
Full grid + seeds & 10,000 & $19 \times 5$ & ${\sim}19$min & 4GB & 8-core CPU, 32GB \\
Scalability test & 100,000 & 1 & ${\sim}25$min & 12GB & 8-core CPU, 32GB \\
\bottomrule
\end{tabular}
\end{table}


\section{Comparison with MIA Resources}\label{sec:appendix-mia}

Direct comparison between CVPL and Membership Inference Attacks (MIA) is methodologically limited since they target different objects (datasets vs.\ models). However, we provide indicative resource requirements for context.

\begin{table}[htbp]
\centering
\caption{Resource comparison: CVPL vs.\ related privacy assessment methods.}
\label{tab:mia-comparison}
\begin{tabular}{@{}llllll@{}}
\toprule
\textbf{Method} & \textbf{Target} & \textbf{Setup} & \textbf{Runtime (10K)} & \textbf{GPU} & \textbf{Reference} \\
\midrule
MIA (shadow models) & Model & Hours & Minutes & Typical & \cite{shokri2017membership} \\
MIA (memorization) & Model & None & Minutes & Varies & \cite{carlini2019secret,carlini2021extracting} \\
\textbf{CVPL} & Dataset & None & $<4$ min & No & This work \\
Fellegi--Sunter & Dataset & Minutes & Minutes & No & \cite{fellegi1969theory} \\
DCR/NNDR & Dataset & None & Seconds & No & \cite{jordon2022synthetic} \\
\bottomrule
\end{tabular}
\end{table}

\paragraph{Caveat.}
MIA complexity varies substantially with attack variant, target model architecture, and available access. CVPL's primary advantage is applicability: it evaluates dataset releases directly, without requiring a trained model.


\section{Extension to Real Data}\label{sec:appendix-real-data}

\subsection{Adult Census Validation (Proposed)}\label{sec:adult-validation}

For future validation, we propose applying CVPL to the UCI Adult Census dataset~\cite{uci_adult}:
\begin{enumerate}
    \item \textbf{Original:} 32,561 records, 14 attributes
    \item \textbf{QIs for protection:} age, education, occupation, native-country
    \item \textbf{Protection:} k-anonymity~\cite{sweeney2002} via deterministic generalization applied per record ($k \in \{5, 10, 20\}$)
    \item \textbf{Ground truth:} Known 1-1 correspondence (same individuals before/after protection)
\end{enumerate}

\paragraph{Expected findings.}
\begin{itemize}
    \item CVPL-LR $> 0$ even for $k = 20$ due to preserved income-education correlations
    \item Feature attribution dominated by workclass and hours-per-week (non-QI)
    \item Fellegi--Sunter~\cite{fellegi1969theory} degradation under heavy generalization
\end{itemize}

\subsection{Protocol for Audit Without Ground Truth}\label{sec:audit-protocol}

When true linkages are unavailable:
\begin{enumerate}
    \item \textbf{Self-linkage diagnostic:} Compute leave-one-out $\CVPLLR(\Dpr, \Dpr)$ with identity matches excluded. High values (relative to chosen $\optau$) confirm the representation has discriminative power. Note: this is a sanity check, not a guarantee of 1.0.
    \item \textbf{Risk screening:} Report $\CVPLLR(\Dor, \Dpr)$ as an \emph{estimate} of linkage feasibility under the CVPL adversary model. This is not a formal upper bound, as real attackers may have capabilities outside the model.
    \item \textbf{Mechanism comparison:} Relative ordering $\CVPLLR(\Dor, \Dpr_A) < \CVPLLR(\Dor, \Dpr_B)$ indicates that mechanism A provides better protection under CVPL assumptions, even without absolute calibration.
    \item \textbf{Feature attribution:} Identify the highest-contributing features for targeted remediation or additional protection.
\end{enumerate}


\section{Glossary}\label{sec:appendix-glossary}

\begin{table}[htbp]
\centering
\caption{Glossary of terms and notation.}
\label{tab:glossary}
\begin{tabular}{@{}lp{10cm}@{}}
\toprule
\textbf{Term} & \textbf{Definition} \\
\midrule
CVPL & Cluster-Vector-Projection Linkage -- the proposed framework \\
CVPL-LR & CVPL Linkage Rate -- primary metric measuring existential linkability (\cref{sec:appendix-consistency}) \\
$R_{\text{recall}}$ & True-match recall -- probability that the correct match exceeds threshold (\cref{sec:appendix-recall}) \\
$\optau$ (tau) & Similarity threshold defining attacker strictness \\
$\widehat{\optau}$ & Calibrated threshold (\eg for false-positive control at level $\alpha$) \\
$\lambda$ (lambda) & Blocking granularity parameter \\
$\rho$ (rho) & Correlation retention factor for synthetic data generation~\cite{xu2019modeling} \\
Blocking & Partitioning records by QI values to constrain candidate sets \\
Candidate set $\mathcal{C}(x)$ & Set of protected records sharing a block with original record $x$ \\
QI & Quasi-identifier -- attribute observable to and usable by attacker~\cite{sweeney2002} \\
Existential linkage & At least one candidate exceeds the similarity threshold \\
Top-1 linkage & The highest-similarity candidate is the true match \\
Risk surface $\RiskSurf{\lambda}{\optau}$ & Linkage risk as a joint function of blocking and threshold \\
$\Splus$ & Similarity distribution for true-match pairs \\
$\Sminus$ & Similarity distribution for non-match pairs \\
$\ystar$ & True match -- the protected record corresponding to original $x$ \\
Precision@1 & Fraction of top-1 correct predictions \\
DCR & Distance to Closest Record -- nearest-neighbor distance metric~\cite{jordon2022synthetic} \\
NNDR & Nearest Neighbor Distance Ratio -- uniqueness metric~\cite{jordon2022synthetic} \\
RCE & Record Closest to Existing -- memorization diagnostic (introduced in this work) \\
FS & Fellegi--Sunter -- classical probabilistic record linkage model~\cite{fellegi1969theory} \\
\bottomrule
\end{tabular}
\end{table}

\bibliographystyle{plain}
\bibliography{bib/references}

\end{document}